\documentclass[a4paper,unpublished,onecolumn,11pt]{quantumarticle}
\pdfoutput=1
\usepackage[utf8]{inputenc}
\usepackage[english]{babel}
\usepackage[T1]{fontenc}
\usepackage{amsmath}
\usepackage{amssymb}
\usepackage{physics}

\usepackage{hyperref}
\usepackage[table]{xcolor}
\usepackage{colortbl}
\usepackage{booktabs}
\usepackage{enumerate}
\usepackage{enumitem}
\usepackage{amsthm}
\usepackage{caption}
\usepackage{subcaption}

\usepackage{selinput}
\SelectInputMappings{
  adieresis={ä}, 
  acute={´},
}
\usepackage{amsmath}  

\usepackage[english]{babel}
\usepackage[T1]{fontenc}
\usepackage{lipsum}
\usepackage{mdframed}
\usepackage{booktabs}
\usepackage{adjustbox}
\usepackage{longtable}

\usepackage[
  backend=biber,
  style=numeric,
  sorting=none      
]{biblatex}
\usepackage{xcolor}
\usepackage[table]{xcolor}
\usepackage{booktabs}
\usepackage{listings}
\usepackage{siunitx}
\usepackage{multicol}
\lstdefinestyle{TinyPython}{
  language        = Python,
  basicstyle      = \ttfamily\fontsize{5pt}{6pt}\selectfont,  
  numbers         = left,
  numberstyle     = \ttfamily\fontsize{5pt}{6pt}\selectfont,
  stepnumber      = 1,
  numbersep       = 3pt,
  breaklines      = true,
  breakatwhitespace = true,
  showstringspaces = false,
  frame           = single,
  xleftmargin     = 0pt,
  xrightmargin    = 0pt,
  aboveskip       = 0pt,
  belowskip       = 0pt,
  tabsize         = 2,
  columns         = flexible,
  keepspaces      = true,
}

\usepackage{pifont}  
\usepackage{newunicodechar}
\newunicodechar{✖}{\ding{55}}
\usepackage[margin=1in]{geometry}   
\usepackage{setspace}              
\usepackage{amssymb} 
\usepackage{bm}
\usepackage{amsmath}
\usepackage{amsthm}  
\usepackage{tikz}            
\usepackage{quantikz} 
\usepackage{enumitem}
\usepackage{setspace}              
\usepackage{bm}
\usepackage{amsmath,amssymb,amsthm}             
\usepackage{tikz}            
\usepackage{quantikz} 
\usepackage{enumitem}
\usepackage{graphicx} 
\usepackage{algpseudocode}
\usepackage{algorithm}
\usepackage{algorithmicx}
\usepackage{amsthm}
\usepackage{xcolor}
\usepackage[most]{tcolorbox}
\usepackage{pgfplots}
\pgfplotsset{compat=1.17}
\newcommand{\poly}[1]{\mathrm{poly}\bigl(#1\bigr)}

\usepackage{booktabs}
\usepackage{float}                 
\usepackage{hyperref}              
\usepackage{url}
\usepackage{xcolor}                
\usepackage{lscape}                
\usepackage{caption}
\usepackage{subcaption}
\usepackage{etoolbox}
\AtBeginEnvironment{abstract}{\raggedright}
\usepackage{amsmath,amssymb,amsthm}
\newtheorem{theorem}{Theorem}
\newtheorem{lemma}[theorem]{Lemma}
\newtheorem{proposition}[theorem]{Proposition}
\newtheorem{corollary}[theorem]{Corollary}
\newtheorem{definition}[theorem]{Definition}
\DeclareUnicodeCharacter{0301}{'} 

\theoremstyle{remark}
\newtheorem{remark}[theorem]{Remark}

\usepackage{pgfplots}
\usepgfplotslibrary{fillbetween}

\newcommand{\OH}{\mathcal{H}_{\mathrm{OH}}}

\DeclareMathOperator*{\E}{\mathbb{E}}   

\newcommand{\Hpen}{H_{\mathrm{pen}}}
\newcommand{\Hobj}{H_{\mathrm{obj}}}
\newcommand{\UM}{U_{M}}
\newcommand{\UC}{U_{C}}

\makeatletter
\g@addto@macro\normalsize{%
    \abovedisplayskip 3pt plus 1pt minus 1pt%
    \abovedisplayshortskip 3pt plus 1pt minus 1pt%
    \belowdisplayskip 3pt plus 1pt minus 1pt%
    \belowdisplayshortskip 3pt plus 1pt minus 1pt%
}
\makeatother

\def\BibTeX{{\rm B\kern-.05em{\sc i\kern-.025em b}\kern-.08em
    T\kern-.1667em\lower.7ex\hbox{E}\kern-.125emX}}

\author{Chinonso Onah}
\affiliation{Volkswagen AG, Berliner Ring 2, Wolfsburg 38440, Germany}
\affiliation{Department of Physics, RWTH Aachen, Germany}

\author{Roman Firt}
\affiliation{Volkswagen AG, Berliner Ring 2, Wolfsburg 38440, Germany}

\author{Kristel Michielsen}
\affiliation{Department of Physics, RWTH Aachen, Germany}
\affiliation{Forschungszentrum Jülich, Germany}

\title{Empirical Quantum Advantage in Constrained Optimization from Encoded Unitary Designs}

\begin{document}
\maketitle

\begin{abstract}
\noindent
We introduce the \emph{Constraint–Enhanced Quantum Approximate Optimization Algorithm} (CE–QAOA), a shallow, constraint-aware ansatz that operates \emph{inside} the one–hot product space \(\OH=[n]^m\), where \(m\) is the number of blocks and each block is initialized in an \(n\)-qubit \(W_n\) state. We give an ancilla-free, depth-optimal encoder that prepares \(W_n\) using \(n{-}1\) two-qubit rotations per block, and a two-local block-XY mixer that preserves the one–hot manifold and has a constant spectral gap on the one-excitation sector. At the level of expressivity, we establish per-block controllability, implying approximate universality per block. At the level of distributional behavior, we show that (after natural block/symbol permutation twirls) shallow CE–QAOA realizes an encoded unitary \(1\)-design and supports approximate second-moment (\(2\)-design) behavior; combined with a Paley--Zygmund argument, this yields finite-shot anticoncentration guarantees. Algorithmically, we wrap constant-depth sampling with a deterministic feasibility checker to obtain a polynomial-time hybrid quantum--classical solver (PHQC) that returns the best observed feasible solution in \(O(Sn^2)\) time, where \(S=\poly n\) is the shot budget. We obtain two advantages. First, when CE--QAOA fixes \(r\ge 1\) locations different from the start city, we achieve a \(\Theta(n^{r})\) reduction in shot complexity even against a classical sampler that draws uniformly from the feasible set. Second, against a classical baseline restricted to raw bitstring sampling, we show an \(\exp(\Theta(n^{2}))\) minimax separation. In noiseless circuit simulations of TSP instances with \(n\in\{4,\dots,10\}\) locations from the QOPTLib benchmark library, we recover the global optimum at \(p{=}1\) using polynomial shot budgets and coarse parameter grids defined by the problem size, as suggested by Theorem~\ref{thm:exist-params} and Algorithm~\ref{alg:PHQC1}.
\end{abstract}

\footnote{Data availability:
A minimal Python implementation and integration into standard QAOA wrapper in Qiskit is made available \url{https://doi.org/10.5281/zenodo.15725265}}%

\section{Introduction}
\label{sec:introduction}
\noindent
Combinatorial optimization problems (COP) underpin modern logistics, manufacturing, finance, and drug discovery; and it surfaces throughout artificial intelligence and machine-learning workflows~\cite{PapadimitriouSteiglitz1982,PadbergRinaldi1991BranchCut,TothVigo2014VRP,BengioLodiProuvost2021}. Classical methods like branch--and--bound and cutting planes~\cite{PadbergRinaldi1991BranchCut,LawlerWood1966BranchBound,LinKernighan1973} excel when they exploit problem structure. In contrast, many near-term quantum approaches—notably QAOA~\cite{Farhi2014QAOA} and its variants~\cite{montanezbarrera2024universalqaoa,BaeLee2024RecursiveQAOA,Finzgar2024QIRO}—tend to explore the full Hilbert space in search of the (near) optimal solution(s), paying in \emph{measurement cost}~\cite{Cerezo2021VQAReview,Tilly2022VQEReview}, \emph{depth}~\cite{Cerezo2021VQAReview,Tilly2022VQEReview}, and \emph{barren plateaus}~\cite{McClean2018BarrenPlateaus}. A growing line of work mitigates this by encoding directly into symmetry-rich feasible subspaces and using feasibility-preserving mixers~\cite{Hadfield2019AOA,Fuchs2022ConstrainedMixers,BaertschiEidenbenz2020,tsvelikhovskiy2024equivariant,tsvelikhovskiy2024symmetries,Xie2024CVRP}. In this work we study \emph{Constraint--Enhanced Quantum Approximate Optimization Algorithm} (CE--QAOA), a shallow, constraint-aware ansatz that \emph{natively} operates on a block one-hot manifold. Our circuit constructions are based on an ancilla-free, depth-optimal encoder that prepares a block \(W_n\) state using exactly \(n{-}1\) excitation preserving two qubit gates, and a two-local block--XY mixer that preserves the manifold.\cite{Hadfield2019AOA} The properties arising from these circuit primitives allow us to make contact with low-order moment analysis of random quantum circuits\cite{BrandaoHarrowMixing2016} in the encoded subspace and to subsequently study the emergence of approximate unitary $t-$ designs. 


\medskip \noindent
Why do unitary $t$-designs matter? In the context of variational quantum circuits, unitary designs are a principled proxy for \emph{expressivity}, \emph{anticoncentration}, and \emph{typical-case trainability}~\cite{BrandaoHarrowMixing2016}. Recall that a unitary $t$-design reproduces Haar moments up to degree $t$. Thus any degree-$\le t$ observable (e.g., overlaps, energies, and many gradient statistics) computed under the design agrees with the Haar baseline up to~$\varepsilon$~\cite{BrandaoHarrowMixing2016}. Secondly, for $t=2$, 2-designs suggests an \emph{instance-agnostic typical-case behaviour} because second moments are insensitive to the fine details of the cost landscape. From the “landscape information” viewpoint~\cite{PerezSalinasWangBonetMonroig2024}, matching Haar second moments ensures that a nontrivial amount of task-relevant variance survives at shallow depth on the \emph{encoded} space, which increases the effective information capacity of the ansatz~\cite{PerezSalinasWangBonetMonroig2024,HadfieldHoggRieffel2022}. 

\medskip \noindent
Our circuit design inherits from prior work in Refs. \cite{Hadfield2019AOA, doCarmo2025warmstartingqaoa}. Here we leverage additional problem specific structures through problem-algorithm codesign to capture a broad class of COPs (TSP/ATSP, QAP, CVRP, etc) via block one-hot encodings and shared constraint properties. Consequently, we propose a unifying kernel in Def. \ref{def:kernel-requirement}. Within this kernel, a block-permutation twirl yields a unitary \(1\)-design on the encoded space, reproducing the Haar baseline \(\mathbb{E}\,|\langle x|U|\phi\rangle|^2=1/D\) and a corresponding \emph{existence} theorem for any target basis vector \(x^\star\) (and thus for the optimum) in Thm.~\ref{thm:exist-params} and Cor. \ref{cor:feasible-optimum}. Algorithmically, we introduce the the Polynomial time Hybrid Quantum-Classical Solver  (PHQC)  which wraps constant-depth quantum sampling over a coarse  parameter grid with a deterministic classical checker that identifies the best sampled solutions, operationalizing these $t-$design bounds. PHQC identifies the optimal solution if it is sampled at least once.

\medskip \noindent
Our running examples are derived from the Travelling Salesman Problem (TSP) where the \(n^2\) qubits are imagined to be arranged in \(n\) blocks of size \(n\) and the double one-hot constraint restricts to a fixed-Hamming-weight manifold containing only \(n!\) basis states out of the encoded $n^n$ encoded basis vectors and  \(2^{n^2}\) vectors in the ambient Hilbert space~\cite{Lucas2014Ising}. Similar one-hot encoded problems arise naturally for assignment problems~\cite{Edmonds1965Blossom}, matching problems~\cite{Kuhn1955Hungarian}, vehicle routing problems~\cite{TothVigo2014VRP}, graph coloring problems~\cite{Lucas2014Ising} and the recently introduced shared transportation framework~\cite{onah2025waas}. CE--QAOA exploits this structure at the encoder, mixer layers, phase layers, and classical feasibility checker; extending  problem--algorithm co-design from quantum to classical stage~\cite{Li2021CoDesign,Tomesh2021QuantumCodesign}. The advertised quantum advanatge comes from direct utilization of problem structures at the classical and quantum levels of the algorithm.  

\medskip \noindent
We present noiseless circuit simulation results for TSP instances ranging from 4 to 10 locations drawn from the QOPTLib benchmark problems\cite{Osaba2024Qoptlib}. We demonstrate that PHQC recovered the optimal solution in all instances considered while dramatically reducing the shot cost to a polynomial $O(n^{k})$ at depth $p=1$ as reported  in Table \ref{tab:design-bound}. All plots presented in the following can be reproduced as demonstrated in \cite{Onahempdata}. 


\medskip \noindent


\subsection{Relation to Prior Work}
\label{sec:background}
\medskip \noindent
Recently, Smith\mbox{--}Miles \emph{et\,al.}~\cite{smith-miles2025tsp-ntqa} surveyed the obstacles to near-term quantum advantage on TSP—including resource inflation from generic encodings (e.g., QUBO), feasibility handling, noise/optimizer pathologies, and inconsistent benchmarking against strong classical solvers—and argue for structure-exploiting formulations and transparent, reproducible evaluations. Our CE\mbox{--}QAOA/PHQC approach aligns with this perspective but differs in emphasis because we treat global constraints as \emph{useful structure} within an encoded manifold. Our formalism inherits from the preexisiting body of work in the \emph{alternating-operator} literature (often called QAOA+) which develops mixers that commute with, or otherwise preserve, constraint projectors---e.g., one-hot, ring mixers, and null-swaps~\cite{Hadfield2019AOA,Chancellor2017}. For routing families (e.g., CVRP), feasibility is frequently enforced via Grover-style oracles~\cite{Xie2024CVRP,BaertschiEidenbenz2020,LaRose2022MixerPhaser}, which introduce oracle/ancilla overhead and reflection-based subroutines. In contrast, our kernel uses a \emph{two-local, ancilla-free} block--XY mixer that \emph{preserves} all blockwise one-hot constraints, acts entirely within the encoded manifold, and is normalized to have a \emph{constant} spectral gap on the one-excitation sector (Prop.~\ref{prop:spectral-gap}). The initial state is the block-$W$ product\cite{carmo2025warmstarting} prepared by an ancilla-free, depth-optimal cascade (Thm.~\ref{thm:encoder-opt}), making the mixer’s top eigenvector the initial state. Our mixer design ensures \emph{invariance} (excitation conservation) and block factorization, enabling quditization $H_{M}^{(b)} \mapsto A(K_n)$\cite{Hadfield2019AOA} and explicit control of norms and gaps at the encoded per block level. This normalization is used quantitatively in our mixing/expressivity arguments and avoids hidden $n$-dependent angles that can plague unnormalized XY implementations.


\medskip \noindent
Unlike prior feasibility-preserving approaches (e.g.QAOA+~\cite{Hadfield2019AOA, BaertschiEidenbenz2020,LaRose2022MixerPhaser}),
our kernel \emph{explicitly} requires that the \emph{constraints} be exchangeable under block permutations $S_m$ and symbol relabelings $S_n$. This symmetry is a co-design choice that is is enforced at the modeling level and then exploited algorithmically. Two consequences are
central to our analysis: (i) $H_{\mathrm{pen}}$ is invariant under $S_m\times S_n$, so every penalty level set $L_t=\{x:\,H_{\mathrm{pen}}(x) :=\langle x \mid H_{\mathrm{pen}} \mid x \rangle=t\}$ (in particular $L_0$) is preserved setwise; and (ii) a block-permutation twirl acts transitively on the encoded basis and yields an \emph{exact unitary $1$-design} on $\OH$ (Lemma~\ref{lem:perm-twirl}). This exact first-moment identity directly powers our \emph{existence} guarantee for the optimal basis vector $x^\star$ (Cor.~\ref{cor:feasible-optimum}). This symmetry-driven $1$-design baseline, and the resulting existence results, differentiate our co-designed kernel from prior feasibility-preserving approaches.

\medskip \noindent
Although our kernel uses familiar primitives, here we study them under the lens of norms, spectral properties, and typical case behaviors. By systematic use of \emph{low-order moment} methods (\(L_1\), \(L_2\)) together with symmetry and twirling, we expose and utilize: \textbf{(i)} Lie controllability and approximate universality within the one-excitation sectors induced by the block–XY mixer; \textbf{(ii)} \(\varepsilon\)-\emph{approximate unitary 2-designs at shallow depth} arising from per-block design plus diagonal entanglers; \textbf{(iii)} \emph{exponentially fast mixing} and \emph{anticoncentration} on the encoded manifold; and \textbf{(iv)} expressivity that translates into provable lower tails at the \(c/D\) scale. Where $D$ is the encoded dimension and $c>0$ from Thm. \ref{thm:exist-params}. To our knowledge, prior feasibility-preserving or domain-specific QAOA variants do not provide \emph{universal, instance-agnostic} guarantees of this kind. 

\medskip \noindent
Our analyses remain complementary to recent work on \emph{landscape information content}~\cite{PerezSalinasWangBonetMonroig2024}; and analysis of energy/gradient landscapes and alternating-operator structure~\cite{HadfieldHoggRieffel2022}. Our approach equally remains orthogonal to warm-start schedules~\cite{Egger2021Warm,carmo2025warmstarting}. Warm starts can yield strong numerics but provide limited worst-case control and sometimes require non-eigen initializations of the mixer, violating the adiabatic protocol and limiting performance\cite{He2023A}. We avoid classical initialization completely and focus on provable algorithmic performance. 


\section{Codesign and Circuit Primitives for CE-QAOA }
\label{sec:main}


\subsection{Problem-Algorithm Codesign}
Our goal is to identify potential opportunities for provable algorithmic enhancements by taking explicit advantage of \emph{problem structure}. This aligns with the broader \emph{co-design} perspective in quantum computing, where algorithms and platforms are jointly tailored to application constraints and device characteristics~\cite{Li2021CoDesign,Tomesh2021QuantumCodesign}. Here, however, we omit hardware considerations and focus on problem-algorithm co-design. We look beyond \emph{feasibility-preserving} ans\"atze and mixers~\cite{Hadfield2019AOA,Wang2020XYMixers,Fuchs2022ConstrainedMixers,Xie2024CVRP,LaRose2022MixerPhaser} and impose additional symmetry on the constraint structure consistent with the various problem classes, bringing them under a single kernel defined in Def.~\ref{def:kernel-requirement}. As a consequence, performance bounds obtained propagate across a broad class of COPs with minimal tweaks. The definition below fixes what it means to be ``in the kernel'' and will be the used henceforth for all results that follow.

\begin{definition}[CE--QAOA kernel]
\label{def:kernel-requirement}
An optimization instance $I$ belongs to the \emph{CE--QAOA kernel} if there exist
integers $n,m\in\mathbb N$ and the \emph{one-hot} encoder $\mathsf E_{\mathrm{1hot}}$
that initializes the dynamics in the fixed–Hamming–weight space
\[
\OH \;=\; (\mathcal H_1)^{\otimes m},
\qquad
\mathcal H_1 \;=\; \mathrm{span}\{\ket{e_0},\dots,\ket{e_{n-1}}\}
\quad\text{(one excitation per block)}.
\]
The problem Hamiltonian splits as
\[
H_C \;=\; H_{\mathrm{pen}} \;+\; H_{\mathrm{obj}},
\]
Where $H_{\mathrm{obj}}$ is the Ising Hamiltonian representing the  objective and only needs to be diagonal in the computational basis. $H_{\mathrm{pen}}$ is the penalty Hamiltonain and enjoys the symmetries specified in (a) and (b) in addition to being diagonal in the computational basis. 
\begin{enumerate}[label=\textup{(\alph*)}, leftmargin=2.2em]
\item \emph{Penalty structure.} $H_{\mathrm{pen}}$ is a sum of squared affine
      one–hot/degree/capacity penalties (optionally plus linear forbids) with
      integer coefficients bounded by $\mathrm{poly}(n)$. Consequently,
      $\mathrm{spec}(H_{\mathrm{pen}})\subseteq\{0,1,\dots,t_{\max}\}$ with
      $t_{\max}=O(m)=\mathrm{poly}(n)$.
\item \emph{Pattern symmetry.} $H_{\mathrm{pen}}$ is invariant under
      (i) block permutations $S_m$ and (ii) global symbol relabelings $S_n$.
      Hence the configuration space decomposes into level sets
      $L_t=\{x:\, H_{\mathrm{pen}}(x) :=\langle x \mid H_{\mathrm{pen}} \mid x \rangle=t\}$ that are preserved setwise.
\item In addition, the initial state is the $+1$ eigenstate of the block–local normalized XY mixer Hamiltonian,
      \[
      \widetilde H_{M}^{(b)} \;=\;
      \frac{1}{n-1}\sum_{0\le j<k\le n-1}(X_j^{(b)}X_k^{(b)}+Y_j^{(b)}Y_k^{(b)}),
      \]
      with $\|\widetilde H_{M}^{(b)}\|=O(1)$ on each block. The initial state is the       uniform one–hot product
      \[
      \ket{s_0}\;=\;\ket{s_{\mathrm{b}}}^{\otimes m},
      \qquad
      \ket{s_{\mathrm{b}}}\;=\;\frac{1}{\sqrt n}\sum_{k=0}^{n-1} \ket{e_k}
      \quad\text{(a $W_n$ state per block)}.
      \]
\end{enumerate}
\end{definition}

Although the kernel above fixes one-hot blocks ($1$ excitation per block), many applications naturally use \emph{$k$-hot} blocks (fixed Hamming weight of $k$ per block). In that case, one replaces $\mathcal H_1$ by
\(
\mathcal H_k=\mathrm{span}\{\ket{x}\in\{0,1\}^n:\, \|x\|_1=k\}
\),
and chooses a suitable \emph{number-conserving} block-local mixer (e.g.\ modified  XY Hamiltonian, which preserves excitation number) normalized to constant operator norm on $\mathcal H_k$. The corresponding initial state becomes a tensor product of Dicke states encoders~\cite{B_rtschi_2019},
\[
\ket{s_{\mathrm{b}}^{(k)}}\;=\;\ket{D^{(n)}_{k}}
\;=\;\binom{n}{k}^{-\tfrac12}\!\!\sum_{\substack{x\in\{0,1\}^n\\ \|x\|_1=k}}\!\ket{x},
\qquad
\ket{s_0}\;=\;\bigl(\ket{D^{(n)}_{k}}\bigr)^{\otimes m}.
\]
This further enlarges the class of constrained combinatorial optimization problems that can fit into the kernel proposed above. Some identified problem classes include: Travelling Salesman (TSP/ATSP)~\cite{Lawler1985TSP,GareyJohnson1979}, Quadratic Assignment Problem (QAP)~\cite{Koopmans1957QAP,Loiola2007QAPSurvey}, Vehicle Routing with Capacities (CVRP)~\cite{TothVigo2014VRP}, Generalized Assignment / Multiple–Knapsack~\cite{SahniGonzalez1976GAP,MartelloToth1990Knapsack}, $k$D Matching (NP-complete for $k\ge3$)~\cite{Karp1972,garey1979computers}, the shared transportation problems\cite{onah2025waas}, and Job–Shop / Flow–Shop Scheduling~\cite{GareyJohnsonSethi1976,LenstraRinnooyKan1977}. Some base cases (e.g., linear assignment) are polynomial-time solvable via the Hungarian algorithm~\cite{Kuhn1955Hungarian,Munkres1957Hungarian}; nevertheless, they fit the kernel structurally. Our interest is in the NP-hard variants (e.g., TSP/QAP/CVRP, etc.) for which the same one-hot encodings and symmetries apply.

\begin{definition}[Feasible basis vectors]
\label{def:feasible-basis}
Let $\OH=\mathcal H_{1}^{\otimes m}$ denote the encoded one–hot product space from
Definition~\ref{def:kernel-requirement}, and let
\[
H_C \;=\; H_{\mathrm{pen}} + H_{\mathrm{obj}}
\]
be the problem Hamiltonian with $H_{\mathrm{pen}}$ satisfying the kernel symmetries.
A computational basis vector $\ket{x}$ with block decomposition
$x=(j_0,\ldots,j_{m-1})\in [n]^m$ is called \emph{feasible} if and only if
\[
    H_{\mathrm{pen}}(x) :=\langle x \mid H_{\mathrm{pen}} \mid x \rangle \;=\; 0,
\]
i.e.\ $\ket{x}$ lies in the zero–penalty level set
\[
   L_0(H_{\mathrm{pen}})\;:=\;\{\,\ket{x}\in\OH:\ H_{\mathrm{pen}}(x)=0\,\}.
\]

\end{definition}

\begin{remark}
In the standard one–hot encodings for permutation and routing problems, each basis vector $\ket{x}$ can be reshaped into an $n\times n$ binary matrix
$X$ with exactly one $1$ in each row.  A basis vector is feasible precisely when
this reshaped matrix is a permutation matrix:
\[
   X\in S_n
   \qquad\Longleftrightarrow\qquad
   \sum_{k=0}^{n-1} X_{ik}=1\ \text{for all rows $i$}
   \;\text{ and }\;
   \sum_{i=0}^{n-1} X_{ik}=1\ \text{for all columns $k$}.
\]
Thus the feasible basis vectors are exactly those $x\in\OH$ whose blockwise
one–hot assignments are globally collision–free and satisfy all penalty constraints,
i.e.\ the elements of the permutation manifold $L_0(H_{\mathrm{pen}})\subset\OH$.
\end{remark}

\subsection{ The Initial State Preparation Protocol}
\label{sec:init_state}
Consider some $n\times m$--qubit system where $n$ can be considered of a fixed size repeated $m$ times. Suppose each block has \(n\) qubits, and consider a single excitation subspace spanned by
 the equal superposition \(
|s\rangle = \frac{1}{\sqrt{n}} \sum_{k=0}^{n-1} |e_ k\rangle.
\) Over the $m$ blocks, the final state is:
\begin{equation}
  |s_{\mathrm{0}}\rangle \;=\; 
  \bigotimes_{b=0}^{m-1} \left(\frac{1}{\sqrt{n}}\sum_{k=0}^{n-1} \lvert e_k\rangle\right),
  \label{eq:init}
\end{equation}
realising block-wise tensor product of uniform W-states \cite{diker2022,B_rtschi_2019}. Algorithm ~\ref{alg:wblock}, a simple cascade-style algorithm can prepare the initial state described above in Eq \ref{eq:init}; obtaining a single block of \(n\) qubits in a W-state. Then Algorithm ~\ref{alg:product} (\textbf{\texttt{OneHotMultiBlockPrepare}(\(m\))}) uses the single-block procedure on each of the \(m\) blocks to create a W-state per block. The total Hamming weight per bitstring is fixed and the Hamming Geometry is constrained by the location of single excitation contributed by each block. This circuit construction scales optimally in gate depth on a linear array.

\begin{theorem}[Gate–depth optimality on a line]
\label{thm:encoder-opt}
There exists an ancilla-free circuit that prepares $\ket{W_n}$ on a linear array
using $n-1$ two-qubit \emph{excitation-preserving rotations}
\[
U^{(k,k{+}1)}_{M}(\phi)\;:=\;\exp\!\Bigl[-\,i\,\frac{\phi}{2}\,\bigl(X_k X_{k+1}+Y_k Y_{k+1}\bigr)\Bigr],
\]
on each adjacent pair $(k,k{+}1)$, with suitable angles $\{\phi_k\}_{k=0}^{n-1}$. Moreover, any circuit over single-qubit gates and two-qubit gates on a linear array that prepares $\ket{W_n}$ must use at least $n-1$ two-qubit gates.
\end{theorem}

\begin{proof}
Label qubits $0,1,2,\dots,n-1$ on a line and, for each $k=0,\dots,n-1$, consider the bipartition
$\{0,\dots,k\}\,|\,\{k+1,\dots,n-1\}$.
For the uniform $W$-state $\ket{W_n}=\frac1{\sqrt n}\sum_{k=0}^{n-1}\ket{0\cdots 1_k\cdots0}$,
the Schmidt rank across \emph{every} such single-edge cut is $\mathrm{SR}_k(\ket{W_n})=2$.
The product state $\ket{0}^{\otimes n}$ has rank $1$ across every cut. A two-qubit gate that does \emph{not} cross a given cut cannot increase that cut’s rank, while a two-qubit gate that \emph{does} cross it can increase the rank by at most a factor of $2$. Thus each of the $n{-}1$ adjacency cuts must be crossed at least once by a two-qubit gate, so any preparation requires $\ge n{-}1$ two-qubit gates.  Alg. \ref{alg:wblock} gives a construction that achieves the bound.
\end{proof}

\begin{remark}[Hardware-friendly alternatives]
\label{rem:alt}
On platforms with native $\mathrm{iSWAP}^\alpha$ or $\mathrm{fSim}(\theta,\varphi)$, $U^{(k,k{+}1)}_{M}(\phi)$ is implementable in one entangling gate (up to single-qubit $Z$ phases). If $XX{+}YY$ is not native, one can realize the same two-level rotation on the
preparation path with the single 2Q primitive
\[
\tilde U(\theta)\;:=\;\mathrm{CX}_{\,k{+}1\to k}\;\cdot\;\mathrm{CRY}_{\,k\to k{+}1}(\theta)
\quad\text{with}\quad \theta=2\phi,
\]
which has the block action
$\ket{10}\mapsto \cos(\tfrac{\theta}{2})\,\ket{10}+\sin(\tfrac{\theta}{2})\,\ket{01}$ on
$\mathrm{span}\{\ket{10},\ket{01}\}$.
\end{remark}

\begin{algorithm}[H]
\caption{\texttt{OneHotBlockPrepare}(\(n\)) — exact single–excitation $W_n$}
\label{alg:wblock}
\begin{algorithmic}[1]
\Require number of qubits \(n\ge 2\) labelled \(0,\dots,n{-}1\)
\State For adjacent qubits \((k,k{+}1)\), let
\[
\mathrm{U}^{(k,k{+}1)}(\theta)\;:=\;\exp\!\Bigl[-\,i\,\tfrac{\theta}{2}\,\bigl(X_kX_{k+1}+Y_kY_{k+1}\bigr)\Bigr]
\]
\Comment{acts as a real Givens rotation on \(\mathrm{span}\{\ket{10},\ket{01}\}\)}
\Statex \hspace{\algorithmicindent}\emph{(Hardware alternative: }\(\mathrm{U}^{(k,k{+}1)}(\theta)=
\mathrm{CX}_{\,k{+}1\to k}\cdot \mathrm{CRY}_{\,k\to k{+}1}(\theta)\)\emph{)}
\State Apply \(X\) on qubit \(0\) \Comment{\(\ket{100\ldots0}\)}
\For{\(k \gets 0\) to \(n-2\)}
    \State \(\displaystyle \theta_k \gets 2\arccos\!\Bigl(\frac{1}{\sqrt{\,n-k\,}}\Bigr)\)
    \State Apply \(\mathrm{U}^{(k,k{+}1)}(\theta_k)\) \Comment{fixes amplitude on site \(k\) to \(1/\sqrt n\)}
\EndFor
\Ensure \(\displaystyle \ket{W_n} = \tfrac1{\sqrt n}\sum_{k=0}^{n-1}\ket{0\ldots1_k\ldots0}\)
\end{algorithmic}
\end{algorithm}

\begin{algorithm}[H]
\caption{\texttt{OneHotMultiBlockPrepare}(\(m\)) — tensor product of $m$ $W$-blocks}
\label{alg:product}
\begin{algorithmic}[1]
\Require block size $n$, number of blocks $m$
\For{$b \gets 0$ to $m-1$}
    \State Apply \textsc{OneHotBlockPrepare}$(n)$ on qubits
           $\;bn,\;bn{+}1,\ldots,(b{+}1)n-1$
\EndFor
\Ensure $\displaystyle
  \bigl(\ket{W_n}\bigr)^{\!\otimes m}
  \;=\;
  \frac{1}{\sqrt{n^m}}
  \sum_{\substack{k_0,\ldots,k_{m-1}\\\in\{0,\ldots,n-1\}}}
  \bigotimes_{b=0}^{m-1} \ket{0\ldots 1_{k_b} \ldots 0}$
\end{algorithmic}
\end{algorithm}

\begin{figure}[h]
\centering
\includegraphics[width=.89\linewidth]{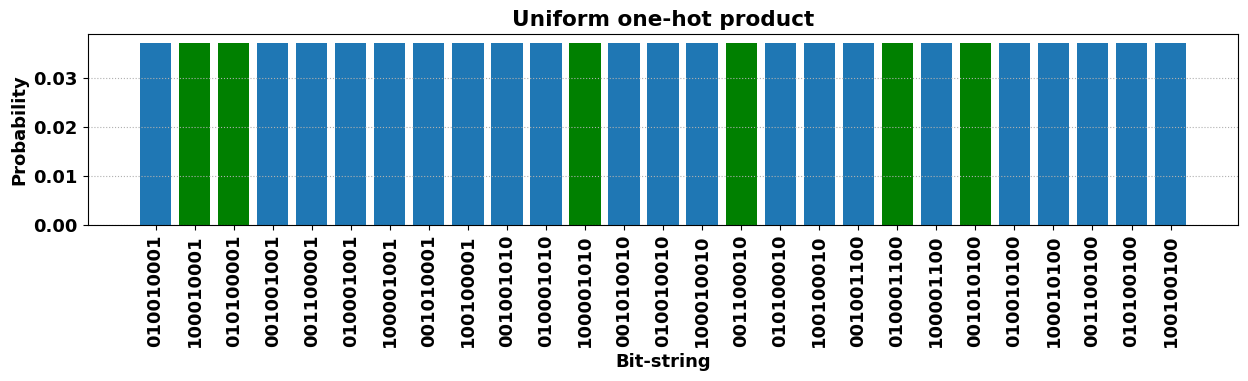}

\caption{Histogram produced in \textsc{Qiskit}\cite{Qiskit2023} for a TSP problem on 3 locations. Each of the \(3^{3}=27\) block–one-hot bit-strings carries probability
\(1/27\) verifying that the state–preparation algorithm populates all basis states with equal amplitude while leaving probability mass strictly inside the one-hot subspace. The six strings whose block indices form a permutation are coloured green.}
\label{fig:uniform-27}
\end{figure}

\noindent
It follows from well known results in quantum circuit simulation that the initial state prepared by Algorithm ~\ref{alg:product} is non-stabilizer for $n\ge3$\cite{bravyi_gosset_2016} because each block contains a $W_n$-state. We cite this result in the following proposition.
\begin{proposition}[${W_n}$ is non-Clifford\;\cite{bravyi_gosset_2016}]%
\label{prop:W_non_stabiliser1}
For $n\ge 3$ the $n$-qubit $W$ state
$\ket{W_n}=\frac{1}{\sqrt n}\sum_{k=0}^{n-1} \ket{e_k}$
is not a stabilizer state. In particular, it cannot be prepared from
$\ket{0}^{\otimes n}$ using Clifford operations alone. (The edge case $W_2=(\ket{01}+\ket{10})/\sqrt2$ is stabilizer.) Outside the stabilizer regime, strong simulation typically relies on non-Clifford techniques which are not known to run in polynomial time in general~\cite{bravyi_gosset_2016}.
\end{proposition}

\noindent
It is thus reasonable to expect that our encoder already places the ansatz beyond the efficiently simulable stabilizer regime, strengthening the potential for quantum–classical separation discussed in Sections~\ref{sec:classical-competitiveness} and \ref{sec:discussion}. Fig.\ref{fig:uniform-27} shows the basis vectors forming a product of one-hot states in the reduced Hilbert space with bitstrings representing valid permutations marked in green. Fig. \ref{fig:full-Blockqaoa} is a single layer CE-QAOA circuit for a $12$ qubit problem with $3$ blocks\cite{Onahempdata}. \footnote{Data availability:
A minimal Python implementation and integration into standard QAOA wrapper in Qiskit is made available \url{https://doi.org/10.5281/zenodo.15725265}}%





\subsection{The XY Mixer Operator}
Recall that the task in QAOA+ is to construct, for the initial state, a suitable  mixer Hamiltonian such that the initial state is its $+1$ eigenstate\cite{tsvelikhovskiy2024symmetries, tsvelikhovskiy2024equivariant}. Thus, we require an XY mixer Hamiltonian that preserves the one--hot subspace. For each n--qubit block \(b\) we define  (up to constant prefactor) 
\begin{equation}
\label{eq:mixer}
H_M^{(b)} = \sum_{0\leq i < j \leq n-1}\Bigl( X_{i}^{(b)}X_{j}^{(b)}+Y_{i}^{(b)}Y_{j}^{(b)}\Bigr),
\end{equation}

where \(X_i^{(b)}\) and \(Y_i^{(b)}\) denote the Pauli \(X\) and \(Y\) operators acting on the \(i\)th qubit in block \(b\). The overall mixer unitary is given by
\begin{equation}
U_M(\beta) =  U_M^{(0)}(\beta) \otimes \dots \otimes U_M^{(m-1)}(\beta).
\end{equation}


Where  \(
U_M^{(b)}(\beta) = e^{-i\beta\, H_M^{(b)}}
\) for each block. The construction is outlined in Algorithm~\ref{alg:mixer}. See Fig. \ref{fig:full-Blockqaoa} for the blockwise construction depicted in a full circuit.  


\medskip
\noindent
In terms of the dynamics, the  XY Mixer continuously mixes the valid one--hot basis states among themselves thanks to its adjacency and ergodicity properties on the  $\mathcal H_1$ subspace.  To expose these properties, it is helpful to rewrite the mixer Hamiltonian in a form that is manifestly 2-local and preserves the excitation number. The mixer Hamiltonian is:

\begin{equation}
\label{eq:mixerb}
    H_{M}^{(b)}   \;:=\;
    \sum_{0\le i<j\le n-1}\!\bigl(X_{i}^{(b)}X_{j}^{(b)}+Y_{i}^{(b)}Y_{j}^{(b)})
    \;=\;
    \sum_{i<j}\!\bigl(\sigma_{bi}^{+}\sigma_{bj}^{-}
                   + \sigma_{bi}^{-}\sigma_{bj}^{+}\bigr),
\end{equation}
Restricting attention to the
one-excitation subspace
\(
   \mathcal H_{1}
   = \operatorname{span}\{\ket{e_0},\dots,\ket{e_{n-1}}\},
\)  and normalising by $n$ gives
\(
   \widetilde H_{M}^{(b)} :=H_{M}^{(b)} /n
\)
with \emph{constant} spectral gap. (cf. App \ref{app:two-local}).

\begin{algorithm}[H]
  \caption{\textsc{BlockXYMixer}\,$(n,m,\beta)$ — two-local mixer for the one-hot subspace}
  \label{alg:mixer}
  \begin{algorithmic}[1]
    \Require Block size $n$, number of blocks $m$, mixer angle $\beta$
    \For{$b \gets 0$ \textbf{to} $m-1$}
      \State \emph{iterate over blocks}
      \ForAll{$(i,j)$ \textbf{with} $0 \le i < j < n-1$}
        \State \emph{apply two-local mixer on qubit pair $(i,j)$ of block $b$}
        \State $\operatorname{RXX}_{bn+i,bn+j}\!\bigl(2\beta\bigr)$
        \State $\operatorname{RYY}_{bn+i,bn+j}\!\bigl(2\beta\bigr)$
      \EndFor
    \EndFor
    \Ensure
      \[
        U_M(\beta)
        = \exp\!\Bigl[-i\beta
          \sum_{b=0}^{m-1}
          \sum_{0\le i<j<n-1}
          \left(
      \frac{1}{n-1}\sum_{0\le j<k\le n-1}(X_j^{(b)}X_k^{(b)}+Y_j^{(b)}Y_k^{(b)})\right)
        \Bigr]
      \]
  \end{algorithmic}
\end{algorithm}

\begin{proposition}[Spectral gap of one-block XY mixer]
\label{prop:spectral-gap}
On $\mathcal H_{1}$ the operator $H_{M}^{(b)} $ acts as the adjacency matrix
$A(K_{n})$ of the complete graph on $n$ vertices and has spectrum
\[
   \operatorname{spec}\bigl(H_{M}^{(b)} \!\upharpoonright_{\mathcal H_{1}}\bigr)
   = \bigl\{\,n-1,\;\underbrace{-1,\dots,-1}_{n-1\text{ times}}\bigr\}.
\]
Hence the spectral gap is $\Delta(H_{M}^{(b)} ) = n$.  
\(
   \widetilde H_{M}^{(b)} =H_{M}^{(b)} /n
\)
has constant gap
\(
   \Delta(\widetilde H_{M}^{(b)} ) = 1.
\)
\end{proposition}

\begin{proof}

Acting on a basis state with the excitation at site~$i$ with Eq. \ref{eq:mixerb} (See App. \ref{app:two-local}),
\[
   (X_iX_j+Y_iY_j)\ket{e_i}
     \;=\;
     \sigma_i^{-}\sigma_j^{+}\ket{e_i}
     \;=\;
     \ket{e_j},
\]
while the same operator annihilates $\ket{e_k}$ for $k\notin\{i,j\}$.
Summing over all unordered pairs $(i,j)$ therefore maps
\(
   \ket{e_i}\mapsto\sum_{j\ne i}\ket{e_j},
\)
which is exactly the action of the adjacency matrix $A(K_n)$ on the
standard vertex basis.  The spectrum of $A(K_n)$ is well known from spectral graph theory~\cite{GodsilRoyle2001}. Largest eigenvalue $\lambda_{\max}=n-1$ (eigenvector
$\frac1{\sqrt n}\sum_{k}\ket{e_k}$) and the remaining $n-1$ eigenvalues equal $-1$.  Hence $\Delta(H_{M}^{(b)} ) = (n-1)-(-1)=n$.  Dividing by $n$ rescales both extremal eigenvalues by the same factor, so $\Delta(\widetilde H_{M}^{(b)} )=1$.
The claimed \emph{constant} spectral gap and $K_n$ denotes its complete adjacency graph. 
\end{proof}

\begin{proposition}[Ergodicity of the angle-averaged XY mixer on $\mathcal H_1$]
\label{prop:xy-ergodicity}
Consider a single $n$-qubit one-hot block with one-excitation sector
$\mathcal H_1=\mathrm{span}\{\ket{e_0},\dots,\ket{e_{n-1}}\}$.
Let $H_{M}^{(b)} \!\upharpoonright_{\mathcal H_1}=A(K_n)$ be the restriction of the all-to-all unnomalised XY Hamiltonian to $\mathcal H_1$ (equivalently, the adjacency matrix of the complete graph on $n$ vertices up to an overall scalar).
For $\beta\in\mathbb R$ define $U(\beta):=e^{-i\beta H_{M}^{(b)} }$ and the transition matrix
\[
P_{ij} := \int_{0}^{2\pi} \frac{\mathrm{d}\beta}{2\pi}\,
\left|\langle e_j \,|\, U(\beta) \,|\, e_i \rangle\right|^{2},
\qquad 0 \le i,j \le n-1.
\]

Then:
\begin{enumerate}
    \item $P$ is \emph{primitive} (all entries are strictly positive), hence the associated Markov chain is \emph{ergodic} (irreducible and aperiodic).
    \item $P$ is \emph{doubly stochastic}, and its unique stationary distribution is the uniform distribution $\pi^\star=(1/n,\dots,1/n)$.
    \item Explicitly,
    \[
    P_{ii}\;=\;1-\frac{2}{n}+\frac{2}{n^2},
    \qquad
    P_{ij}\;=\;\frac{2}{n^2}\quad (j\neq i),
    \]
    so that $P^t\to \mathbf 1\,\pi^{\star\!\top}$ as $t\to\infty$.
\end{enumerate}
These conclusions are invariant under any nonzero rescaling $H_{M}^{(b)} \mapsto c\,H_{M}^{(b)} $, $c\in\mathbb R\setminus\{0\}$.
\end{proposition}

\begin{proof}
On $\mathcal H_1$, $H_{M}^{(b)} $ has spectral decomposition
\(
H_{M}^{(b)} =(n-1)\Pi_s+(-1)\Pi_\perp
\),
where $\Pi_s=\ket{s}\!\bra{s}$ with $\ket{s}=\tfrac{1}{\sqrt n}\sum_{k=0}^{n-1}\ket{e_k}$ and $\Pi_\perp=I-\Pi_s$.
Hence
\[
U(\beta)=e^{-i\beta(n-1)}\Pi_s \;+\; e^{+i\beta}\Pi_\perp.
\]
For basis states,
\(
\langle e_j \,|\, U(\beta) \,|\, e_i \rangle
= e^{-i\beta(n-1)}\frac{1}{n} + e^{i\beta}\big(\delta_{ij}-\tfrac{1}{n}\big).
\)

Averaging $\bigl|\cdot\bigr|^2$ over $\beta\in[0,2\pi)$ removes cross terms (orthogonal Fourier modes), giving
\[
P_{ij}
=\frac{1}{n^2}+\Bigl(\delta_{ij}-\frac{1}{n}\Bigr)^2
=
\begin{cases}
1-\dfrac{2}{n}+\dfrac{2}{n^2}, & j=i,\\[6pt]
\dfrac{2}{n^2}, & j\neq i.
\end{cases}
\]
All entries are strictly positive, so $P$ is primitive (hence ergodic). Row and column sums are $1$ by symmetry, so $P$ is doubly stochastic and the uniform distribution is the unique stationary distribution. Convergence $P^t\to \mathbf 1\,\pi^{\star\!\top}$ follows from primitivity. If $H_{M}^{(b)} $ is rescaled by a $c > 0$, the averaged cross terms still vanish and the same $P$ formula holds. In particular it holds for $\widetilde H_{M}^{(b)} $.
\end{proof}

Local unitary designs make extensive use of qudit language. Since t-designs are central to our methods, following Ref. \cite{Hadfield2019AOA}, we formulate the qudit \(\leftrightarrow\)qubit isometry below.

\begin{definition}[Encoding isometry (\emph{qubit} \(\leftrightarrow\) \emph{qudit})]
\label{def:isometry}
Let \(\mathbb C^n\) be the abstract qudit space with basis \(\{\ket{j}\}_{j=0}^{n-1}\).
Define the isometry
\[
V:\ \mathbb C^n \longrightarrow (\mathbb C^2)^{\otimes n},
\qquad
V\ket{j}\;=\;\ket{e_j}.
\]
Then \(V^\dagger V=I_n\) and \(VV^\dagger\) is the projector onto \(\mathcal H_1\).
\end{definition}

\begin{proposition}[Invariance and quditization of the block–XY mixer]
\label{prop:quditization}
Let
\(
H_{M}^{\mathrm{(b)}}=\sum_{0\le i<j\le n-1}\bigl(X_i^{(b)}X_j^{(b)}+Y_i^{(b)}Y_j^{(b)}\bigr)
=\sum_{i\neq j}\sigma_{bi}^{-}\sigma_{bj}^{+}
\)
on the \(n\) qubits of a block.
Then:
\begin{enumerate}[leftmargin=1.5em]
\item \(\mathcal H_1\) is invariant under \(H_{M}^{\mathrm{(b)}}\) and \(U^{\mathrm{(b)}}(\beta):=e^{-i\beta H_{M}^{\mathrm{(b)}}}\).
\item In the encoded qudit picture,
\[
V^\dagger\,H_{M}^{\mathrm{(b)}}\,V
\;=\;
\sum_{i\neq j}\ket{i}\!\bra{j}
\;=\; A(K_n),
\]
the adjacency matrix of the complete graph \(K_n\). Consequently,
\(
U^{\mathrm{(b)}}(\beta)
= V\,e^{-i\beta\,A(K_n)}\,V^\dagger
\)
on \(\mathcal H_1\).
\end{enumerate}
\end{proposition}

\begin{proof}
For any single–excitation basis vector \(\ket{e_j}\),
\(
\sigma_i^{-}\sigma_j^{+}\ket{e_j}=\ket{e_i}
\)
and \(\sigma_i^{-}\sigma_j^{+}\ket{e_k}=0\) if \(k\notin\{i,j\}\).
Thus \(H_{M}^{\mathrm{(b)}}\) maps \(\ket{e_j}\) to a superposition of \(\{\ket{e_i}\}\),
so \(\mathcal H_1\) and \(U^{\mathrm{(b)}}(\beta)\) are invariant.
Moreover,
\(
V^\dagger \sigma_i^{-}\sigma_j^{+} V
= \ket{i}\!\bra{j},
\quad(i\neq j),
\)
hence the claimed identification with \(A(K_n)\).
\end{proof}

\subsection{Constraint-Enhanced QAOA Protocol}
\label{sec:CE-QAOA }

The next key ingredient is the cost Hamiltonian \(H_C\) that represents the optimization problem and constraints. This is usually the Ising Hamiltonian corresponding to the problem and thus diagonal in the computational basis. The full circuit with \(p\) layers alternates between the cost and mixer unitaries a set of parameters (respectively) \(\vec{\gamma}=(\gamma_1,\dots,\gamma_p)\) and \(\vec{\beta}=(\beta_1,\dots,\beta_p)\). Algebraically, the circuit is represented by the following:

\begin{equation}
|\psi_p(\vec{\gamma},\vec{\beta})\rangle = \left(\prod_{l=1}^{p} e^{-i\beta_l H_M} e^{-i\gamma_l H_C}\right)|s_{\text{0}}\rangle.
\end{equation}
The Alg. \ref{alg:product} prepares $|s_{\text{0}}\rangle$ and initializes the quantum dynamics in the constrained fixed–Hamming–weight space
\(
\OH \;=\; (\mathcal H_1)^{\otimes m} \). 
Because the cost Hamiltonian \(H_C\) is diagonal in the computational basis, we have \(U_C(\gamma)|x\rangle \in \OH\) for all product basis states \(\lvert x\rangle \in \OH\) and it never mixes amplitudes or leaks outside the one-hot subspace.  The overall depth increase from our construction, if any, is minimal. The following holds.

\begin{corollary}[Block size limits circuit depth]
\label{corr:fixed_depth}
Let the block-structured QAOA ansatz consist of $m$ disjoint blocks, each acting on $n$ qubits. Then, under the assumption of parallel execution across blocks, the circuit depth of the full initial state preparation unitary is $\mathcal{D}_\mathrm{prep} = \mathcal{D}_{\mathrm{block\text{-}prep}}$, and the depth of the full mixer unitary is $\mathcal{D}_\mathrm{mix} = \mathcal{D}_{\mathrm{block\text{-}mix}}$. Hence, the overall depth overhead is independent of the number of blocks $m$.
\end{corollary}

\noindent
\textbf{Resource savings from co-design.}
Suprisingly, the co-designed restriction to the one-hot manifold can \emph{reduce} quantum resources relative to the standard QAOA formulation\cite{Farhi2014QAOA}. For permutation-constrained problems (e.g., TSP), feasibility imposes a double one-hot structure (row and column constraints) on the permutation matrix. In CE--QAOA, the chosen initial state and block-XY mixer preserve the row-wise one-hot condition by construction, so those constraint terms are redundant and can be dropped from the diagonal phase operator; only the complementary (e.g., column-wise) constraints need to be enforced. This eliminates an entire family of penalty interactions and yields an \(O(n)\) reduction in the depth (and two-qubit gate count) of the phase-separation unitary compared to standard QAOA implementations. Corollary~\ref{corr:fixed_depth} already shows that the additional overhead from the proposed state preparation and mixing does not scale with the number of blocks \(m\) under parallel execution, so the net effect can be a \emph{strict} depth reduction. For example, the \(O(n)\) savings from removing redundant penalties can dominate the \(O(n)\) preparation cost on architectures with native (or efficiently compiled) excitation-preserving rotations leading to an overall circuit depth smaller than that of a generic QAOA.

\medskip
\noindent
Crucially, these resource savings do not come from limiting the reachable set of feasible states. On the contrary, in the next proposition we show that the same block-local, constraint-preserving operations are controllable (and approximately universal) \emph{within each block}.


\begin{proposition}[Controllability and approximate universality on $\mathcal H_1$]
\label{prop:controllability-universality}
Fix a single $n$-qubit one-hot block and the operator set
\[
\mathcal G \;:=\; \bigl\{\,H_{ij}^{(b)}=X_i^{(b)}X_j^{(b)}+Y_i^{(b)}Y_j^{(b)} \,:\, 0\le i<j\le n-1,0\le b\le m-1\,\bigr\}\;\cup\;\bigl\{\,H_Z^{(1)}\,\bigr\},
\]
where $H_Z^{(1)}$ is any diagonal excitation-preserving term whose restriction to $\mathcal H_1$ is \emph{not} proportional to the identity.
Restricted to the one-excitation sector $\mathcal H_1=\mathrm{span}\{\ket{e_0},\dots,\ket{e_{n-1}}\}$, the dynamical Lie algebra generated by $i\mathcal G$ is $\mathfrak{su}(n)$.
Consequently, for any $V\in \mathrm{SU}(n)$ and any $\varepsilon>0$ there exist angles and a finite product $\prod_{\ell=1}^{L} e^{-i\theta_\ell H_{\alpha_\ell}}$ with $H_{\alpha_\ell}\in\mathcal G$ such that
\[
\Bigl\|\, V - \prod_{\ell=1}^{L} e^{-i\theta_\ell H_{\alpha_\ell}} \,\Bigr\| \;\le\; \varepsilon.
\]
\end{proposition}

\begin{proof}[Proof sketch(See App \ref{app:proof})]
Work in the one-excitation basis $\{\ket{e_k}\}_{k=0}^{n-1}$ of $\mathcal H_1$. On this sector, each XY coupling acts as a symmetric matrix unit,
$H_{ij}\!\mid_{\mathcal H_1}=E_{ij}+E_{ji}$, while the diagonal excitation-preserving term restricts to a non-scalar diagonal
$D=H_Z^{(1)}\!\mid_{\mathcal H_1}=\sum_k d_k E_{kk}$ with $D\not\propto I$. The key observation is that commutators with $D$ turn symmetric off-diagonals into skew-symmetric ones: whenever $d_i\neq d_j$,
\[
[D,\,E_{ij}+E_{ji}] \;=\; (d_i-d_j)(E_{ij}-E_{ji}),
\]
so the Lie closure contains both $E_{ij}+E_{ji}$ and $E_{ij}-E_{ji}$ for at least one pair. Commuting these then produces traceless diagonals,
\[
[E_{ij}+E_{ji},\,E_{ij}-E_{ji}] \;=\; 2(E_{ii}-E_{jj}),
\]
and iterating standard bracket identities yields the full space of traceless diagonal operators on $\mathcal H_1$.

Finally, using any traceless diagonal that separates $k$ and $\ell$,
\[
[E_{kk}-E_{\ell\ell},\,E_{k\ell}+E_{\ell k}] \;=\; 2(E_{k\ell}-E_{\ell k}),
\]
one obtains skew-symmetric off-diagonals for every $(k,\ell)$. Hence the Lie closure contains the usual spanning set of
$\mathfrak{su}(n)$ given by symmetric off-diagonals, skew-symmetric off-diagonals, and traceless diagonals, so
$\mathrm{Lie}(i\mathcal G)\!\mid_{\mathcal H_1}=\mathfrak{su}(n)$. Standard controllability then implies that products of
$\exp(-i\theta H)$ with $H\in\mathcal G$ can approximate any target $V\in \mathrm{SU}(n)$ to arbitrary precision.
\end{proof}

\begin{figure}[htbp]
  \centering
  \resizebox{\textwidth}{!}{%
  \begin{quantikz}[row sep=0.1cm, column sep=0.05cm]
  \lstick[wires=4]{$\text{Block }0\; \ket{0}^{\otimes n}$}
      & \gate[wires=4,style={rounded corners,fill=blue!8}]{\texttt{OneHotBlock}}
      & \qw
      & \gate[wires=12,style={rounded corners,fill=green!16}]{U_C(\gamma_1)}
      & \gate[wires=4,style={rounded corners,fill=orange!20}]{U_M^{(0)}(\beta_1)}
      & \qw
      & \gate[wires=12,style={rounded corners,fill=green!16}]{U_C(\gamma_2)}
      & \gate[wires=4,style={rounded corners,fill=orange!20}]{U_M^{(0)}(\beta_2)}
      & \qw
      & \gate[wires=12,style={rounded corners,fill=green!16}]{U_C(\gamma_3)}
      & \gate[wires=4,style={rounded corners,fill=orange!20}]{U_M^{(0)}(\beta_3)}
      & \qw
      & \meter{} \\
  & & \qw
    & \qw & \qw & \qw
    & \qw & \qw & \qw
    & \qw & \qw & \qw
    & \meter{} \\
  & & \qw
    & \qw & \qw & \qw
    & \qw & \qw & \qw
    & \qw & \qw & \qw
    & \meter{} \\
  & & \qw
    & \qw & \qw & \qw
    & \qw & \qw & \qw
    & \qw & \qw & \qw
    & \meter{} \\
  \lstick[wires=4]{$\text{Block }1\; \ket{0}^{\otimes n}$}
      & \gate[wires=4,style={rounded corners,fill=blue!8}]{\texttt{OneHotBlock}}
      & \qw
      & \qw
      & \gate[wires=4,style={rounded corners,fill=orange!20}]{U_M^{(1)}(\beta_1)}
      & \qw
      & \qw
      & \gate[wires=4,style={rounded corners,fill=orange!20}]{U_M^{(1)}(\beta_2)}
      & \qw
      & \qw
      & \gate[wires=4,style={rounded corners,fill=orange!20}]{U_M^{(1)}(\beta_3)}
      & \qw
      & \meter{} \\
  & & \qw
    & \qw & \qw & \qw
    & \qw & \qw & \qw
    & \qw & \qw & \qw
    & \meter{} \\
  & & \qw
    & \qw & \qw & \qw
    & \qw & \qw & \qw
    & \qw & \qw & \qw
    & \meter{} \\
  & & \qw
    & \qw & \qw & \qw
    & \qw & \qw & \qw
    & \qw & \qw & \qw
    & \meter{} \\
  \lstick[wires=4]{$\text{Block }2\; \ket{0}^{\otimes n}$}
      & \gate[wires=4,style={rounded corners,fill=blue!8}]{\texttt{OneHotBlock}}
      & \qw
      & \qw
      & \gate[wires=4,style={rounded corners,fill=orange!20}]{U_M^{(2)}(\beta_1)}
      & \qw
      & \qw
      & \gate[wires=4,style={rounded corners,fill=orange!20}]{U_M^{(2)}(\beta_2)}
      & \qw
      & \qw
      & \gate[wires=4,style={rounded corners,fill=orange!20}]{U_M^{(2)}(\beta_3)}
      & \qw
      & \meter{} \\
  & & \qw
    & \qw & \qw & \qw
    & \qw & \qw & \qw
    & \qw & \qw & \qw
    & \meter{} \\
  & & \qw
    & \qw & \qw & \qw
    & \qw & \qw & \qw
    & \qw & \qw & \qw
    & \meter{} \\
  & & \qw
    & \qw & \qw & \qw
    & \qw & \qw & \qw
    & \qw & \qw & \qw
    & \meter{} \\
  \end{quantikz}}
  \caption{Depth-$p=3$ CE-QAOA for $m=3$ blocks of $n=4$ qubits.
  Each layer applies a global cost $U_C(\gamma_\ell)$ over all $mn$ wires,
  followed by parallel block-local XY mixers $U_M^{(j)}(\beta_\ell)$.}
  \label{fig:full-Blockqaoa}
\end{figure}
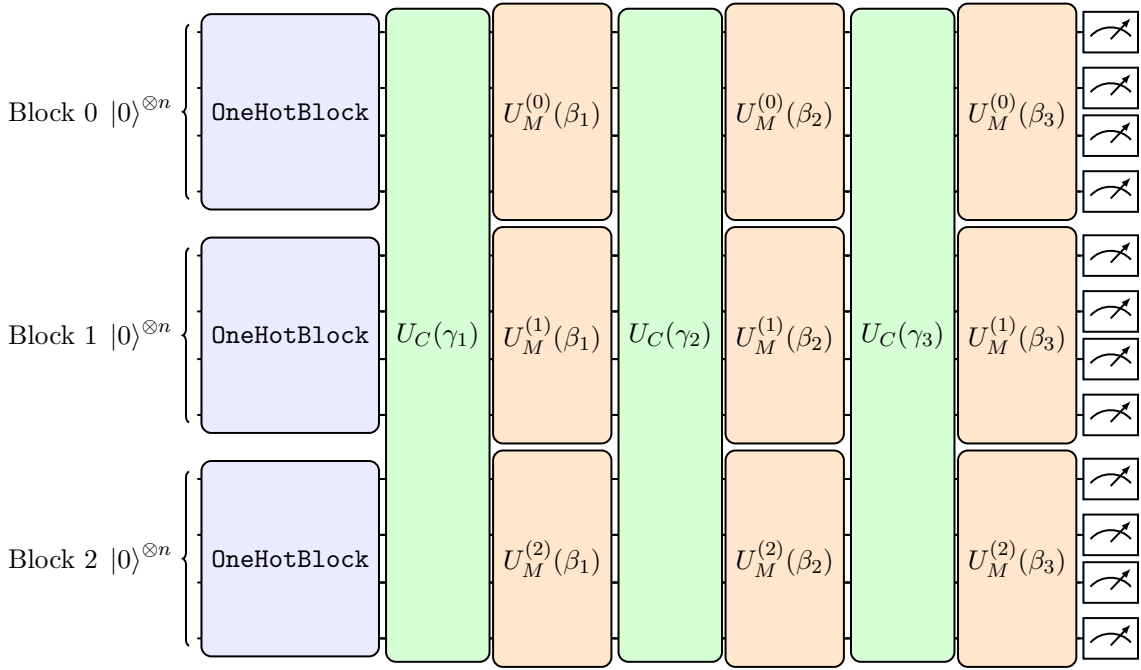

\section{Encoded Unitary Designs}
\label{sec:performance}
\textbf{Roadmap.} In \S2.1 we set up the problem–algorithm co-design and formalized the \emph{CE–QAOA kernel}, fixing the encoded space, symmetry assumptions, mixer normalization, and initial states. In \S3.1 we show that blockwise label permutations induce an \emph{exact unitary 1-design} on the encoded space, yielding a clean \(1/D\) baseline and an averaging-to-existence guarantee for favorable angles (and explaining heavy-output peaks when no twirl is applied). In \S3.2 we lift this to \(\varepsilon\)-approximate \emph{2-design} control via per-block XY controllability combined with diagonal cross-block entanglers provides second-moment (anticoncentration) bounds at poly depth. In \S3.3 we benchmark against classical procedures under two ambient domains, isolating the \(D/S\) vs.\ \(2^{n^2}/S\) baselines and the resulting conditional separation. Finally, \S3.4 packages these ingredients into a \emph{polytime hybrid quantum–classical} (PHQC) solver that couples a coarse \((\gamma,\beta)\) grid with a deterministic checker, yielding Chernoff-style shot guarantees independent of empirical frequency.

\subsection{Exact Unitary \texorpdfstring{$1-$}-design from block permutation twirl}
\label{app:1-design}
We now turn to low order moment analysis~\cite{BrandaoHarrowMixing2016} starting from the emergence of exact unitary $1$ design in CE-QAOA. Our kernel  (Def. \ref{def:kernel-requirement}) \emph{explicitly} requires that the \emph{constraints} be exchangeable under block permutations $S_m$ and symbol relabelings $S_n$. This symmetry is a co-design choice that is enforced at the modeling level and then exploited algorithmically contributing a very rich dynamical structure to our formalism. For example, the block Level permutation–twirling action preserves the one-hot condition on the Hilbrt space $\OH$ (See App. \ref{app:perm} for a global feasibility preserving permutation). To see this, let
\[
\mathcal B \;\equiv\; S_n^{\times m} \;=\; \Bigl\{\;\mathsf P=\bigotimes_{b=0}^{m-1} \mathsf P_b \;:\; \mathsf P_b\in S_n\;\Bigr\}
\;\subset\; U(\OH),
\]
which applies \emph{independent} label permutations in each block. This action is \emph{not} feasibility–preserving in general (it can break column one–hot), but it does implement a first–moment (\(1\)-design) average on the encoded space as an immediate consequence. 

\medskip
\noindent
Consider a coarse uniform grid search with \(n{+}1\) points:
\[
  \mathcal G_{n+1} := \left\{\, j\frac{\pi}{n} : j=0,1,\dots,n \right\}.
\]
and take \(\gamma\in\mathcal G_{n+1}\) and \(\beta\in\mathcal G_{n+1}\). Let \(\widetilde H_{M}^{(b)} \) be the normalized block-XY mixer on the one-excitation subspace. The following existence result holds:  


\begin{theorem}[Existence of Uniform Overlap Lower Bound via Block Permutations]
\label{thm:exist-params}
For any product basis vector  $\ket{x^\star}\in \OH$, and any angles
\((\gamma,\beta)\),
there exists a blockwise permutation \(\mathsf P^\star\) such that
\[
  \bigl|\langle x^{\star}\!\mid
          \mathsf P^{\star\dagger} U_{M}(\beta)
          U_{C}(\gamma)
          s_{0}\rangle
  \bigr|^{2}
  \;\ge\;
  \frac{c}{n^{m}}.
\]
\end{theorem}

To prove this, we shall need two Lemmas. 
\begin{lemma}[Existence from an average (pigeonhole/averaging principle)]
\label{lem:existence-avg}
Let $\{O_t\}_{t\in T}$ be a finite family of Hermitian operators on any Hilbert space
$\mathcal H$, and let $\rho$ be any density operator (positive semidefinite, trace $1$) on $\mathcal H$.
Define the arithmetic mean operator
\[
  \overline O \;:=\; \frac{1}{|T|}\sum_{t\in T} O_t .
\]
Then there exists an index $t^\star\in T$ such that the expectation value of $O_{t^\star}$ in the
state $\rho$ is at least the expectation value of the mean operator:
\[
  \Tr(\rho\,O_{t^\star}) \;\ge\; \Tr(\rho\,\overline O).
\]
i.e. among the set $\{\,\Tr(\rho\,O_t)\,:\,t\in T\,\}$ at least one is not smaller
than their average.
\end{lemma}

\begin{proof}
By linearity of the trace,
\[
  \frac{1}{|T|}\sum_{t\in T} \Tr(\rho\,O_t)
  \;=\; \Tr\!\left(\rho\,\frac{1}{|T|}\sum_{t\in T} O_t\right)
  \;=\; \Tr(\rho\,\overline O).
\]
If every $\Tr(\rho\,O_t)$ were strictly smaller than $\Tr(\rho\,\overline O)$, their average would also be strictly smaller—contradiction. Hence some $t^\star$ attains at least the average.
\end{proof}

\begin{lemma}[Blockwise twirl (\(\mathcal B\)) yields the $1/D$ baseline]
\label{lem:perm-twirl}
Let $\mathcal H_1\cong\mathbb C^n$ be the one–hot subspace on a block $\mathcal H_1$. Let the product one-hot subspace be \(\OH =\mathcal H_1^{\otimes m}\). So $D=\dim \OH  =n^m$. For each block $b\in\{0,\dots,m-1\}$, let $\mathsf P_b$ be a uniformly random permutation matrix acting on $\{\ket{e_0},\ldots,\ket{e_{n-1}}\}$ and set $\mathsf P:=\bigotimes_{b=0}^{m-1} \mathsf P_b$. Then, for any fixed unitary $U$ on $\OH$, any product basis vector $\ket{x^\star}$, and any state $\ket{\phi}\in\OH$ (all independent of $\mathsf P$),
\begin{equation}
\label{eq:perm-scalar}
\mathbb{E}_{\mathsf P}\!\left[\,\bigl|\langle x^\star \mid \mathsf P^\dagger U \mid \phi \rangle\bigr|^2\,\right]
\;=\;\frac{1}{D}.
\end{equation}
Equivalently, at the operator level,
\begin{equation}
\label{eq:perm-operator}
\mathbb{E}_{\mathsf P}\!\left[\,U^\dagger\,\mathsf P\,\ket{x^\star}\!\bra{x^\star}\,\mathsf P^\dagger\,U\,\right]
\;=\;\frac{I_D}{D}.
\end{equation}
where $I_D$ is the identity on $\OH$.
\end{lemma}

\begin{proof}
Write $x^\star=(j_0^\star,\ldots,j_{m-1}^\star)$ with $\ket{x^\star}=\bigotimes_{b=0}^{m-1} \ket{e_{j_b^\star}}$. For a single block $b$, uniform conjugation by $\mathsf P_b$ averages a rank-one projector over all basis labels:
\[
\mathbb{E}_{\mathsf P_b}\!\left[\mathsf P_b\,\ket{e_{j_b^\star}}\!\bra{e_{j_b^\star}}\,\mathsf P_b^\dagger\right]
=\frac{1}{n}\sum_{j=0}^{n-1} \ket{e_j}\!\bra{e_j}
=\frac{I_n}{n}.
\]
Independence across blocks gives
\[
\mathbb{E}_{\mathsf P}\!\left[\mathsf P\,\ket{x^\star}\!\bra{x^\star}\,\mathsf P^\dagger\right]
=\bigotimes_{b=0}^{m-1} \frac{I_n}{n}
=\frac{I_D}{D}.
\]
Conjugating by the fixed unitary $U$ yields \eqref{eq:perm-operator}. Taking the expectation value of \eqref{eq:perm-operator} in the state $\ket{\phi}$ gives the scalar identity \eqref{eq:perm-scalar}:
\[
\mathbb{E}_{\mathsf P}\!\left[\,\bigl|\langle x^\star \mid \mathsf P^\dagger U \mid \phi \rangle\bigr|^2\,\right]
=\operatorname{Tr}\!\left[U\ket{\phi}\!\bra{\phi}U^\dagger\,\frac{I_D}{D}\right]
=\frac{1}{D}.
\]
\end{proof}


\begin{proof}[Proof of Thm.~\ref{thm:exist-params}]
Fix any angles $(\gamma,\beta)$ and write
\(
U:=U_M(\beta)U_C(\gamma).
\)
Let $\rho:=\ket{s_0}\!\bra{s_0}$ and, for each blockwise permutation
$\mathsf P=\bigotimes_{b=0}^{m-1} \mathsf P_b$, define the observable
\(
O_{\mathsf P}:=U^\dagger\,\mathsf P\,\ket{x^\star}\!\bra{x^\star}\,\mathsf P^\dagger\,U.
\)
By Lemma~\ref{lem:perm-twirl},
\[
\mathbb{E}_{\mathsf P}\!\left[\,O_{\mathsf P}\,\right]
=\frac{I_D}{D}\,,
\qquad D=n^m,
\]
and therefore
\[
\mathbb{E}_{\mathsf P}\!\left[\,\Tr\!\bigl(\rho\,O_{\mathsf P}\bigr)\,\right]
=\Tr\!\Bigl(\rho\,\frac{I_D}{D}\Bigr)
=\frac{1}{D}.
\]
Applying Lemma~\ref{lem:existence-avg} (Existence via averaging) to the finite set
$T=\{\mathsf P\}$ with observables $\{O_{\mathsf P}\}_{\mathsf P\in T}$ and state $\rho$
yields the existence of some blockwise permutation $\mathsf P^\star$ (possibly depending on
$\gamma,\beta,x^\star$) such that
\[
\left|\,\langle x^\star \mid \mathsf P^{\star\dagger} U \mid s_0 \rangle\,\right|^{2}
= \operatorname{Tr}\!\left(\rho\, O_{\mathsf P^\star}\right)
\;\ge\; \frac{1}{D}
= \frac{1}{n^{m}}.
\]

\end{proof}

\begin{corollary}[Feasible–optimum specialization of Thm.~\ref{thm:exist-params}]
\label{cor:feasible-optimum}
Fix any grid angles $(\gamma,\beta)\in\mathcal G_{n+1}\times\mathcal G_{n+1}$ and set
$U=\UM(\beta)\UC(\gamma)$ on $\OH$ with $D=\dim\OH=n^m$.
Let $x^\star$ be any \emph{feasible optimal} label.
Then there exists a blockwise permutation $\mathsf P^\star$ such that
\[
  \bigl|\langle x^\star \mid \mathsf P^{\star\dagger} U \mid s_0 \rangle\bigr|^2
  \;\ge\; \frac{1}{D} \;=\; \frac{1}{n^m}.
\]
In particular, the constant in Theorem~\ref{thm:exist-params} can be taken as $c=1$.
\end{corollary}

\begin{proof}
Theorem~\ref{thm:exist-params} holds for \emph{any} product basis vector $x^\star\in\OH$. It also holds for any grid angles $(\gamma,\beta)\in\mathcal G_{n+1}\times\mathcal G_{n+1}$. Finally, choosing $x^\star$ to be any feasible optimum preserves the bound. Hence the inequality holds with $c=1$.
\end{proof}

\medskip \noindent
Lemmas~\ref{lem:existence-avg}–\ref{lem:perm-twirl} establish a permutation–twirl control: for any fixed $U$ on $\OH$ and any label $|x\rangle$, the averaged success probability equals the $1/D$ baseline. This identity certifies that random relabeling erases instance structure but offers a robust theoretical floor. For fixed angles $U:=U_M(\beta)U_C(\gamma)$ and target $\ket{x^\star}$, define
\[
p_{\gamma,\beta}^{x^\star}(\mathsf P)
\;:=\;
\bigl|\langle x^\star \mid \mathsf P^\dagger U \mid s_0\rangle\bigr|^2 .
\]
The blockwise permutation twirl gives the control identity
$\mathbb E_{\mathsf P}\!\big[p_{\gamma,\beta}^{x^\star}(\mathsf P)\big]=1/D$.
Hence the native labeling $\mathsf P=\mathrm{id}$ realizes one point in this distribution and, by averaging, there must exist labelings with $p_{\gamma,\beta}^{x^\star}(\mathsf P)\!>\!1/D$
and others with $<\!1/D$. Running the circuit \emph{without} random per-shot relabelings
(i.e., without implementing the twirl) samples $p_{\gamma,\beta}^{x^\star}(\mathrm{id})$,
which can be \emph{strictly larger} than $1/D$ whenever the angles and the instance structure produce constructive interference for the native labels. This arises when energetically favored blocks capture probability mass above $\frac{1}{n}$ average. Empirically this appears as peaks for the optimum/near-optimum bars that lie well above the dashed $1/D$ line (cf.\  Figure~\ref{fig:grid-search}).

\medskip
\noindent
By contrast, applying random blockwise permutations per shot \emph{flattens} the histogram toward
the $1/D$ control and typically reduces such peaks. Thus non-twirled sampling is the appropriate
mode for optimization (preserves useful bias), while twirling serves as a null/control for analytical bounds rather
than a performance booster. Consequently, we optimize in the native labeling (no twirl) and interpret any deviation $p_{\mathrm{id}}(x)-1/D$ as certified problem-induced bias. The practical value of such bias is quantified by shot complexity: if the optimal label satisfies $p_\star=\alpha/D$ with $\alpha>1$, the shots required to observe it with high probability shrink by the same factor. The following corollary holds. 

\begin{corollary}[Shot complexity gain from heavy outputs]
\label{cor:shots}
Let $p_\star:=p_{\mathrm{id}}(x^\star)$ for the optimal label $x^\star$ at some grid point $(\gamma,\beta)$.
With $S$ shots, the failure probability to observe $x^\star$ at least once is
$(1-p_\star)^S\le e^{-p_\star S}$. To achieve success probability $\ge 1-\delta$ it suffices that
\[
S\ \ge\ \frac{\ln(1/\delta)}{p_\star}.
\]
At the permutation twirling baseline $p_\star=1/D$, this is $S=\Theta(D\log(1/\delta))$.
Any heavy-output bias $p_\star=c/D$ with $c>1$ yields an $c$-fold reduction
in shots. If $p_\star\ge n^{-k}$ for some fixed $k$ while $D=n^m$ (with $m=n$ in TSP encodings), then $S=\tilde O(n^{k})$ versus $\tilde\Omega(n^{m})$ at baseline.  See Lemma \ref{lem:chernoff}, Thm. \ref{thm:sample_complexity} and Alg. \ref{alg:PHQC1}. The two flavors of computational advantages arising from our proposals are discussed in Sec. \ref{sec:classical-competitiveness} and \ref{sec:discussion}. An illustration of the effects of heavy outputs on success probability is shown in Fig. \ref{fig:heavy}\cite{Onahempdata}. \footnote{Data availability: Minimal Python implementation and integration into the standard QAOA wrapper in Qiskit is made available \url{https://doi.org/10.5281/zenodo.15725265}.} 

\end{corollary}

\usetikzlibrary{arrows.meta,intersections}
\usepgfplotslibrary{fillbetween}
\pgfplotsset{compat=1.17}

\def\yminval{1e-12}

\begin{figure}[ht]
  \centering
\begin{tikzpicture}
  \begin{axis}[
    width=13cm, height=7.2cm,
    xlabel={Problem size $n$}, ylabel={Single-shot success $p$ (log scale)},
    ymode=log, ymin=\yminval, ymax=1,
    xmin=6, xmax=60,
    grid=both, minor grid style={gray!15}, major grid style={gray!30},
    legend style={draw=none, fill=none, font=\small, at={(0.02,0.02)}, anchor=south west},
    tick align=outside, label style={font=\small}, ticklabel style={font=\small},
    samples=400, domain=6:60,
  ]

  \def\kexp{3}

  \addplot[name path=top, draw=none, forget plot]{1};
  \addplot[name path=bottom, draw=none, forget plot]{\yminval};

  \addplot[name path=df_path, draw=none, forget plot]{x^(-\kexp)};

  \addplot[name path=base_path, draw=none, forget plot]{exp(ln(10)*(-0.010*x^2 - 0.20*x))};

  \addplot[red!25, fill opacity=0.25, forget plot]     fill between[of=bottom and base_path];
  \addplot[yellow!45, fill opacity=0.25, forget plot]  fill between[of=base_path and df_path];
  \addplot[green!35, fill opacity=0.22, forget plot]   fill between[of=top and df_path];

  \addplot[ultra thick, blue, forget plot]
    {(0.85 + 1.2*(1 - exp(-0.06*(x-9)))) * x^(-\kexp)};

  \addplot[very thick, green!45!black, dotted, forget plot]{x^(-\kexp)};

  \addplot[very thick, red, dotted, forget plot]{exp(ln(10)*(-0.010*x^2 - 0.20*x))};

  \addlegendimage{ultra thick, blue}
  \addlegendentry{\textcolor{blue}{CE--QAOA heavy outputs ($p =c(n)/D\approx n^{-k}$)}}

  \addlegendimage{very thick, green!45!black, dotted}
  \addlegendentry{\textcolor{green!45!black}{Dimension-free boundary $p=n^{-k}$}}

  \addlegendimage{very thick, red, dotted}
  \addlegendentry{\textcolor{red}{Uniform baseline ($1/D$)}}

  \node[anchor=west, align=left, font=\footnotesize, text=green!40!black]
    at (axis cs: 10, {3.5*(9)^(-\kexp)})
    {\bfseries Quantum Advantage Region\\[-1pt]
     ($p \ge n^{-k}$ $\Rightarrow$ poly shots)};

  \addplot[only marks, mark=*, mark size=1.8pt, blue, forget plot]
    coordinates {(12.25,{12.25^(-\kexp)})};

  \draw[->, thick, blue!90!black]
    (axis cs: 52, {exp(ln(10)*(-0.010*32^2 - 0.20*32))})
    -- (axis cs: 47, {1.15 * 27^(-\kexp)});
  \node[anchor=west, align=left, font=\footnotesize, text=blue!70!black]
    at (axis cs: 27.3, {4.28 * 27^(-\kexp)})
    {\textbf{Heavy outputs lead to $c(n)/D$ $\approx$ $n^{-k}$}};

  \end{axis}
\end{tikzpicture}
  \caption{Heavy outputs vs. baselines. The red dotted line is the twirling baseline of  \(\mathbb{E}\,|\langle x|U|\phi\rangle|^2=1/D\). The light yellow region denotes  \(\mathbb{E}\,|\langle x|U|\phi\rangle|^2>c/D\) existence results from Thm. \ref{thm:exist-params} and Cor. \ref{cor:feasible-optimum}. The green region indicate \(\mathbb{E}\,|\langle x|U|\phi\rangle|^2 \approx n^{-k}\) for a fixed constant $k$. Here we write the constant numerator in Thm. \ref{thm:exist-params} $c$ as $c(n)$ to emphasize that it is instance dependent. It becomes large when constructive interference arises in absence of the permutation twirling. These higher values overlap with $n^{-k}$ region shaded in green. If the probability of the optimum approaches $n^{-k}$, then to achieve success probability $\ge 1-\delta$ it suffices that
\(S\ \ge\ \frac{\ln(1/\delta)}{p_\star}\) as outline in Cor. \ref{cor:shots}. Our simulation results with \(\ln(1/\delta)=10\) is reported in Fig. \ref{fig:heavy-outputs} where the instance dependent heavy outputs  lead to overlap probabilities in the shaded green region. }
  \label{fig:heavy}
\end{figure}
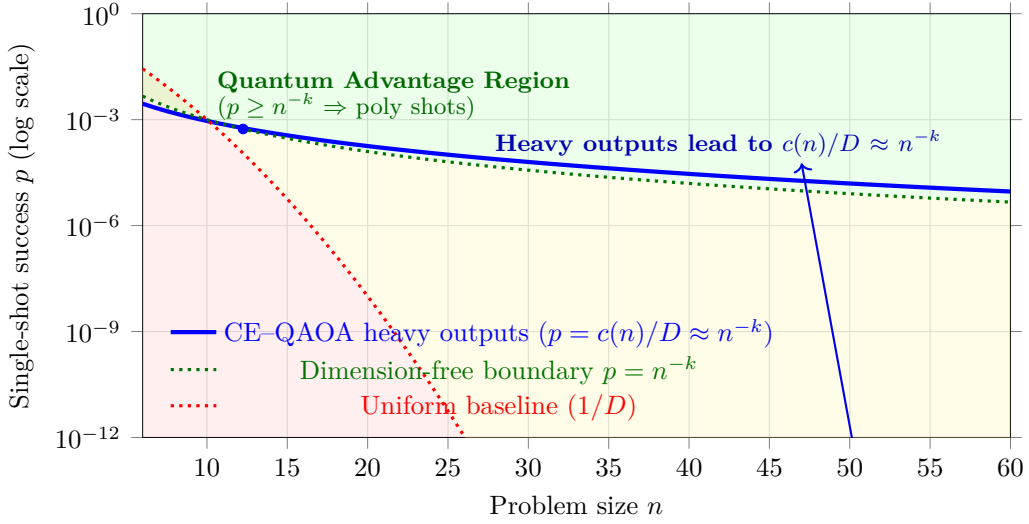

\begin{figure}[ht]
  \centering

  \subfloat[\texttt{wi6}, $(\gamma,\beta)=(\gamma^\star,\beta^\star)$]{
    \includegraphics[width=.45\linewidth]{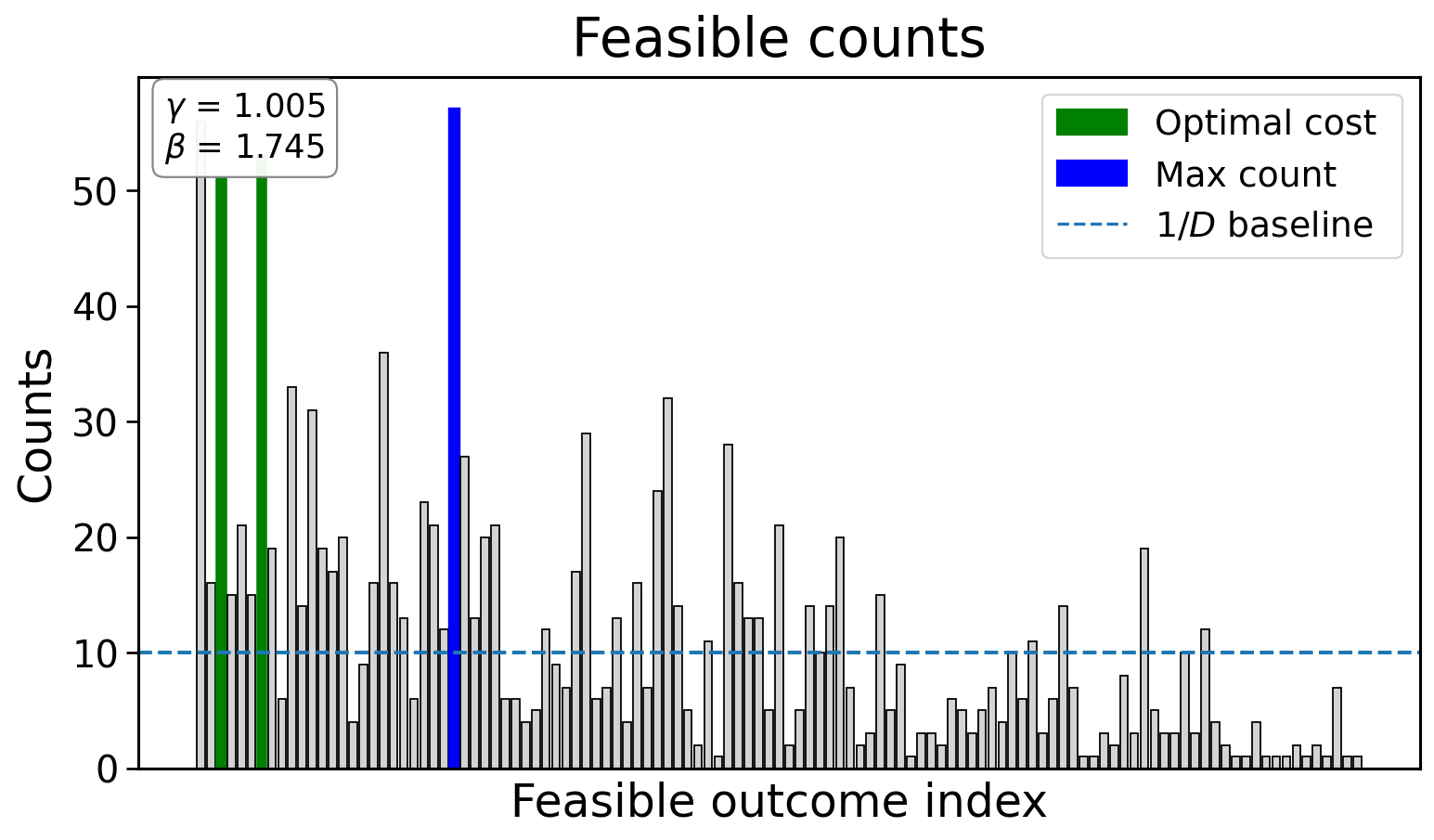}}\hfill
  \subfloat[\texttt{wi6}, $(\gamma,\beta)=(\tilde\gamma,\tilde\beta)$]{
    \includegraphics[width=.45\linewidth]{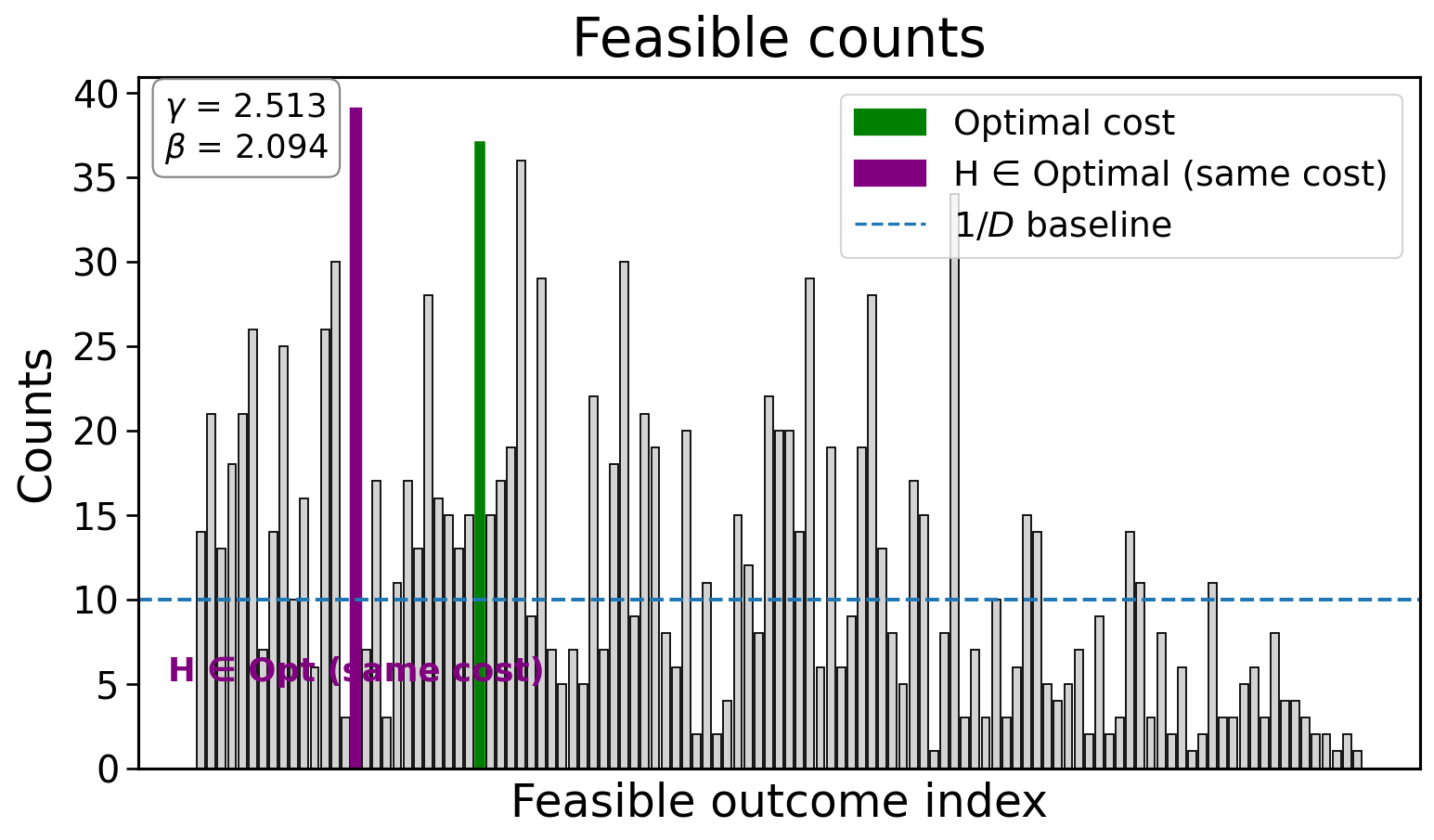}}\\[-4pt]

  \subfloat[\texttt{wi5}, $(\gamma,\beta)=(1.35,0.70)$]{
    \includegraphics[width=.46\linewidth]{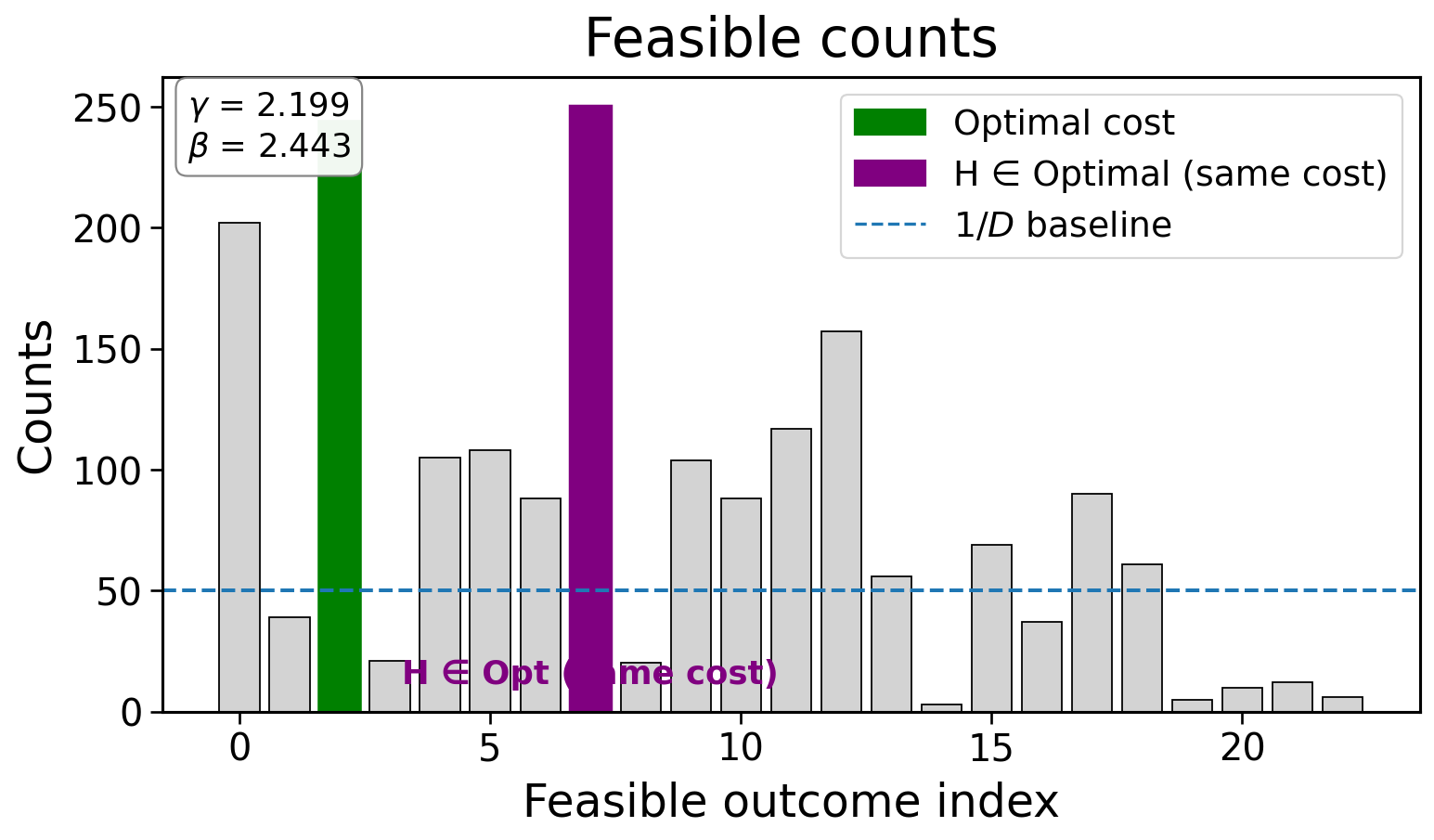}}\hfill
  \subfloat[\texttt{wi5}, $(\gamma,\beta)=(1.35,4.54)$]{
    \includegraphics[width=.46\linewidth]{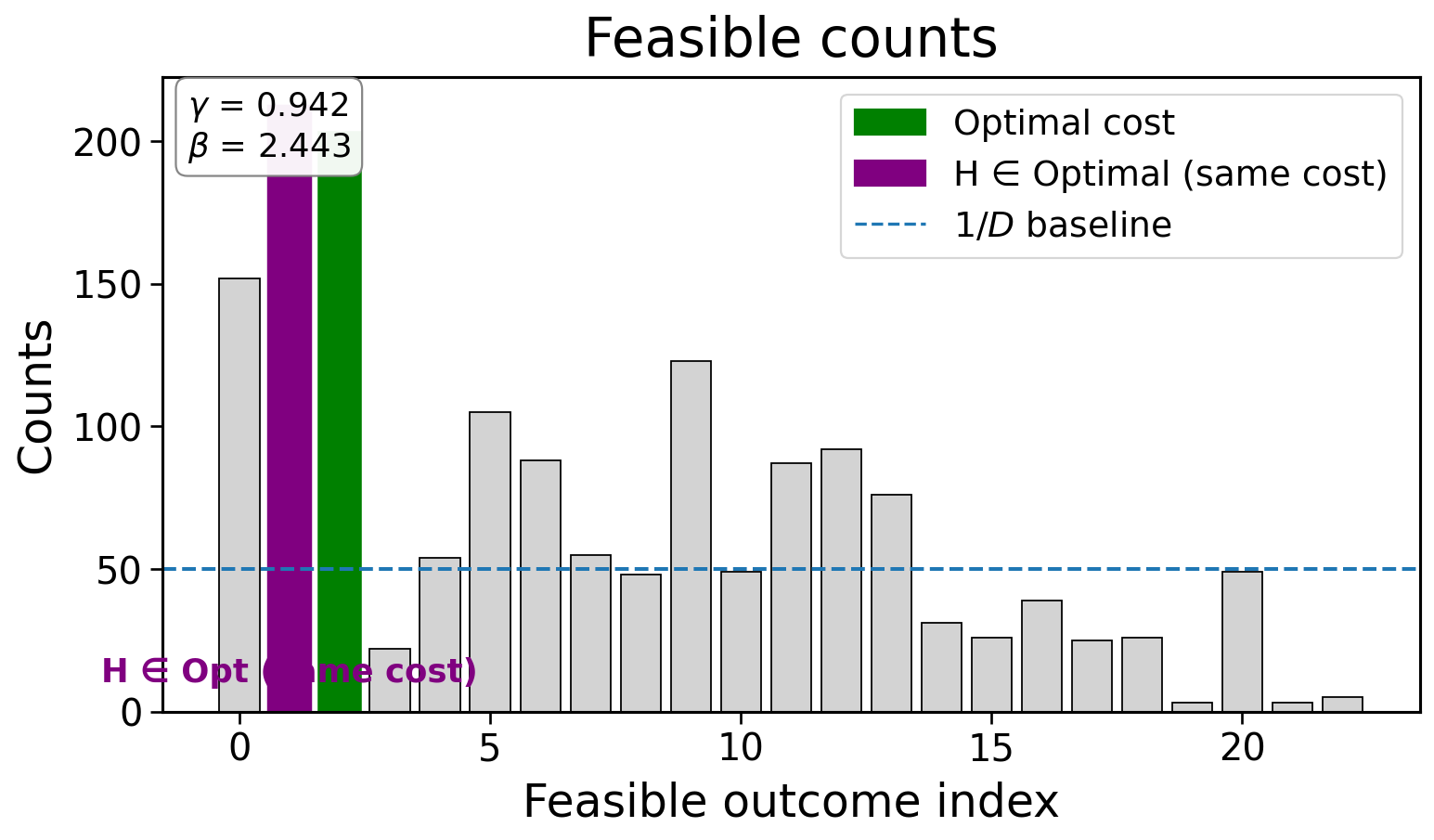}}\\[-4pt]

  \caption{\textbf{Single Layer CE-QAOA circuits on noiseless \textsc{Aer} simulator \cite{Qiskit2023} on two TSP instances ( \texttt{wi5} (25~qubits) and \texttt{wi6} (36 qubits)) taken from QOptLib benchmark set~\cite{Osaba2024Qoptlib}.}  Histograms show the counts of \emph{feasible} bit-strings (permutation indices) out of $50\times n^{4}$ shots for a problem of size $n$ locations. The green bars mark two degenerate best tours. Several grid points already identify the optimum with nontrivial probability well above the lower bound. The dashed line marks the \(1/D\) baseline with \(D=n^n\). The purple bars show when the best bitsring is also the most frequent.}
  \label{fig:grid-search}
\end{figure}

\paragraph{Remarks on the $[0,\pi]^2$ window.}
We use a $\pi$-sized search box in Thm. \ref{thm:exist-params} and the accompanying twirling analyses purely for \emph{coverage}. The $t\!\in\!\{1,2\}$ moment guarantees (perm-twirl and shallow anticoncentration) only require angles drawn from a constant-size set and do not strengthen by widening the range. Empirically, $[0,\pi]^2$ balances exploration of distinct phase patterns with modest grid sizes.  Noiseless simulation results on Qiskit \textsc{Aer} are reported in Figure~\ref{fig:grid-search} to demonstrate this.

\medskip
\noindent
To conclude on the emergence of exact unitary 1-design in our formalism, recall that an ensemble \(\{p_i,U_i\}\subset U(d)\) is a unitary $1$-design if its first-moment (twirling) channel
\[
  \Phi_{\mathrm{Ist}}(X)
  \;:=\;\sum_i p_i\,(U_i)\,X\,(U_i^\dagger)
\]
agrees with the Haar first moment \(\Phi_{\mathrm{Haar}}^{(1)}(X):=\int_{\mathrm{Haar}}(U)\,X\,(U^\dagger)\,dU\)
for all \(X\) on \(\mathbb{C}^d\). It is an \(\varepsilon\)-\emph{approximate} $1$-design if
\(\|\Phi_{\mathrm{Ist}}-\Phi_{\mathrm{Haar}}^{(1)}\|\le\varepsilon\) in diamond (or operator) norm. Thus, the identity proved in Lemma \ref{lem:perm-twirl} is an exact 1-design identity with $\varepsilon =0$. This is a consequence of the block–XY mixer acting as coherent \emph{hops} between labels inside the \(\mathcal H_1\) subspace. In the encoded \(n\)-level (qudit) picture, these hops are realized coherently by the unitary operator, \(e^{-i\beta A(K_n)}\), where \(A(K_n)\) is the adjacency of the complete graph on \(n\) nodes (See Def. \ref{def:isometry} and Prop. \ref{prop:quditization}) and any blockwise permutation operator respects this symmetry. Further t-unitary design details are available for example in Ref. \cite{BrandaoHarrowMixing2016}.

\subsection{\texorpdfstring{$\varepsilon$}{ε}-approximate unitary \texorpdfstring{$2$} {2}- designs}

\medskip
\noindent
Similarly to the 1-design notion, an ensemble \(\{p_i,U_i\}\subset U(d)\) is a unitary $2$-design if its second-moment (twirling) channel
\[
  \Phi_{\mathrm{2nd}}(X)
  \;:=\;\sum_i p_i\,(U_i\otimes U_i)\,X\,(U_i^\dagger\otimes U_i^\dagger)
\]
agrees with the Haar second moment \(\Phi_{\mathrm{Haar}}^{(2)}(X):=\int_{\mathrm{Haar}}(U\otimes U)\,X\,(U^\dagger\otimes U^\dagger)\,dU\)
for all \(X\) on \(\mathbb{C}^d\otimes\mathbb{C}^d\).
It is an \(\varepsilon\)-\emph{approximate} $2$-design if
\(\|\Phi_{\mathrm{2nd}}-\Phi_{\mathrm{Haar}}^{(2)}\|\le\varepsilon\) in diamond (or operator) norm. In our formalism, if on each block’s one–excitation space \(\mathcal H_1\cong\mathbb C^n\), we form the local ensemble by applying
random XY edge pulses, i.e., layers that choose pairs \((i,j)\) and angles \(\theta\) and apply
\(\exp\!\bigl(-i\theta (X_iX_j+Y_iY_j)\bigr)\). Varying the edges and angles yields a
(controllable) gate set whose dynamical Lie algebra is \(\mathfrak{su}(n)\) (Prop. \ref{prop:controllability-universality}). So by local random–circuit design results\cite{BrandaoHarrowMixing2016}, a circuit of length \(L=\mathrm{poly}(n)\cdot\mathrm{polylog}(1/\varepsilon)\)
forms an \(\varepsilon\)-approximate unitary \(2\)-design on \(U(n)\). We capture this result in the following proposition.

\begin{proposition}[Per–block \(\varepsilon\)-approximate $2$–design]
\label{prop:block-2design}
Recall controllability and approximate universality on $\mathcal H_1$ (Prop.\,\ref{prop:controllability-universality}). Let \(\mathcal{U}\) be generated by XY pulses \(\exp(-i\theta (X_iX_j+Y_iY_j))\) and $Z$ phases \(\exp(-i\phi Z_k)\), with \((\theta,\phi)\) drawn from bounded continuous densities.
On \(\mathcal{H}_1\), a random circuit of length \(\mathrm{poly}(n)\cdot \mathrm{polylog}(1/\varepsilon)\) sampled from \(\mathcal{U}\) forms an \(\varepsilon\)-approximate unitary $2$–design \cite{BrandaoHarrowMixing2016}.
\end{proposition}

Interleaving such per–block layers with the diagonal, instance–dependent phase \(U_C(\gamma)\) which is entangling across blocks is needed to  promote the product of per–block designs to a global
\(\varepsilon_{\mathrm{glob}}\)-approximate unitary \(2\)-design on the encoded space in depth
\(\mathrm{poly}(m,n,\log(1/\varepsilon_{\mathrm{glob}}))\).

\medskip
\noindent
In the CE-QAOA  kernel (Def.\,\ref{def:kernel-requirement}), each phase layer is generated by a \emph{diagonal} Hamiltonian in the computational basis which is \emph{entangling across blocks} due to the global pattern symmetries. With block labels \(z=(j_0,\dots,j_{m-1})\), \(U_C(\gamma)\) factorizes over blocks (and thus  non-entangling across blocks ) iff
\begin{equation}
\label{eq:additive-C}
  C(j_0,\dots,j_{m-1})
  \;=\; \sum_{b=0}^{m-1} c_b(j_b)
  \quad\text{for some}\quad c_b:[n]\to\mathbb{R}.
\end{equation}

\begin{proposition}[Diagonal entangler criterion]
\label{prop:diag-entangler}
For two blocks \(b\neq b'\), restrict \(C\) to \(C_{b,b'}(j,k)\) with other labels fixed and consider
\[
  U_{b,b'}(\gamma)
  \;=\;\sum_{j,k=0}^{n-1} e^{-i\gamma C_{b,b'}(j,k)}\,|j\rangle\!\langle j|\otimes|k\rangle\!\langle k|.
\]
The operator–Schmidt rank of \(U_{b,b'}(\gamma)\) exceeds \(1\) (hence the gate is entangling) unless there exist functions \(\alpha(j),\beta(k)\) and a phase \(\theta\) with
\(e^{-i\gamma C_{b,b'}(j,k)}=e^{-i\theta}\,a(j)b(k)\) for all \(j,k\), i.e.\ unless
\(C_{b,b'}(j,k)=\alpha(j)+\beta(k)\) modulo \(2\pi/\gamma\). Equivalently, the diagonal two–qudit gate has operator–Schmidt rank \(>1\) iff the matrix
\(M_{jk}=e^{-i\gamma C_{b,b'}(j,k)}\) has rank \(>1\).
\end{proposition}

\begin{proof}
For a diagonal two–qudit operator, the operator–Schmidt rank equals the rank of the matrix
\(M_{jk}=e^{-i\gamma C_{b,b'}(j,k)}\). Rank \(1\) holds iff \(M_{jk}=a(j)b(k)\).
\end{proof}

Global constraints in typical COPs (TSP/VRP, assignment, capacity/balance) introduce cross–block couplings, so \(U_C(\gamma)\) is \emph{generically} entangling on \((\mathbb{C}^n)^{\otimes m}\). We can therefore combine it with short, per–block $2$–designs to obtain strong second–moment mixing on the encoded space. To that end, we partition the diagonal generator as
\[
H_C \;=\; H_{\mathrm{obj}} + H_{\mathrm{pen}},
\]
with the following roles.
(i) \(H_{\mathrm{obj}}\) is diagonal and, upon restriction to a single block, provides sufficiently rich on–block phases so that, together with XY hops inside the block, the Lie algebra on \(\mathcal H_1\cong\mathbb C^n\) is \(\mathfrak{u}(n)\). (These on–block phases may be present explicitly or synthesized via standard refocusing.)
(ii) \(H_{\mathrm{pen}}\) is diagonal and \emph{cross–block}. It contains at least one non–additive term across some pair \(b\neq b'\), so that for fixed spectators the induced table \(C_{b,b'}(j,k)\) fails \(C_{b,b'}(j,k)=\alpha(j)+\beta(k)\).
By Prop.\,\ref{prop:diag-entangler}, \(e^{-i\gamma H_{\mathrm{pen}}}\) is then entangling across blocks.
(Equivalently, the matrix \(M_{jk}=e^{-i\gamma C_{b,b'}(j,k)}\) has rank \(>1\).)

\medskip
\noindent
If each block ensemble is an \(\varepsilon\)-approximate unitary $2$–design on \(U(n)\), then the product ensemble on \(U(n)^{\otimes m}\) reproduces the Haar second moments of degree–\(\le2\) polynomials with error \(O(m\varepsilon)\).

\begin{theorem}[Global \(\varepsilon_{\mathrm{glob}}\)-approximate unitary $2$–design (mild randomness)]
\label{thm:global-2design}
Assume the block–interaction graph induced by \(H_{\mathrm{pen}}\) is connected and that either
(i) the phases \(\gamma_\ell\) are i.i.d.\ from a bounded continuous density, or
(ii) the set of coupled block pairs in \(H_{\mathrm{pen}}\) is chosen i.i.d.\ each layer from a distribution with full support on the edges of the graph.
Then there exists
\[
  L \;=\; \mathrm{poly}\!\bigl(m,n,\log(1/\varepsilon_{\mathrm{glob}})\bigr)
\]
s.t. the distribution of \(U_L=\mathcal{L}_L\cdots\mathcal{L}_1\) is an \(\varepsilon_{\mathrm{glob}}\)-approximate unitary $2$–design on \(U(n^{m})\).
\end{theorem}

\begin{proof}

\emph{Setup:} Consider alternating layers on \((\mathbb{C}^n)^{\otimes m}\)
\begin{equation}
\label{eq:optionB-layer}
  \mathcal{L}_\ell
  \;=\;
  \left(\bigotimes_{b=0}^{m-1} R_b^{(\ell)}\right)\; e^{-i\gamma_\ell H_{\mathrm{pen}}},
  \qquad \ell=1,\dots,L,
\end{equation}
where each \(R_b^{(\ell)}\) is drawn independently from a per–block \(\varepsilon\)-approximate $2$–design and \(e^{-i\gamma_\ell H_{\mathrm{pen}}}\) is an entangling diagonal gate across blocks (Prop.\,\ref{prop:diag-entangler}). Assume a mild randomness on the penalty layer and fix one layer $\mathcal L_\ell=\bigotimes_b R_b^{(\ell)}\;e^{-i\gamma_\ell H_{\mathrm{pen}}}$.
For $t=2$ moments, analyze the superoperator
$\mathbb E\!\left[U^{\otimes 2}\otimes\overline U^{\otimes 2}\right]$ of the depth–$L$ circuit. \newline
(1) \emph{Local twirl to $\{I,\mathrm{SWAP}\}$.}
Each $R_b^{(\ell)}$ is drawn from a per–block $\varepsilon$–approximate 2–design, so its second–moment twirl
projects (up to $O(\varepsilon)$) any operator on the two–copy space of block $b$
onto the commutant of $U(n)$, i.e.\ the span of $\{I_b,\mathrm{SWAP}_b\}$.
Thus, after averaging over the $R_b^{(\ell)}$, the state space relevant to second moments
reduces to a classical spin system on $m$ sites with local alphabet $\{I,\mathrm{SWAP}\}$. \newline
(2) \emph{Effect of the entangling diagonal.}
Consider an edge $(b,b')$ coupled by a term in $H_{\mathrm{pen}}$.
Conjugation by $e^{-i\gamma_\ell H_{\mathrm{pen}}}$, restricted to this edge and then followed by the local twirls from step (1),
induces a stochastic update on the pair $(\{I,\mathrm{SWAP}\}_b,\{I,\mathrm{SWAP}\}_{b'})$ that is \emph{nontrivial}
whenever the two–block phase table $C_{b,b'}(j,k)$ is non–additive (Prop.~\ref{prop:diag-entangler}).
Hence each coupled edge implements a mixing step on the reduced alphabet. \newline
(3) \emph{Ergodicity and spectral gap (Cf. Prop. \ref{prop:xy-ergodicity}).}
Assume the block–interaction graph of $H_{\mathrm{pen}}$ is connected and that, per layer,
either the coupled edge or the phase parameter is drawn from a distribution with full support.
Then the induced Markov chain on $\{I,\mathrm{SWAP}\}^m$ is irreducible and aperiodic,
with a spectral gap bounded below by $\Omega(1/\mathrm{poly}(m))$.
Consequently the product of $L=\mathrm{poly}(m,n)\,\mathrm{polylog}(1/\varepsilon_{\mathrm{glob}})$ such layers
converges to the unique fixed point, which coincides with the global Haar second moment on $(\mathbb C^n)^{\otimes m}$
(up to the usual $O(\varepsilon)$ error from the local designs). \newline
(4) \emph{In conclusion,} this is exactly the $t{=}2$ instance of the Brand\~ao--Harrow--Horodecki random–circuit mixing framework\cite{BrandaoHarrowMixing2016},
with the qudit case handled by the same moment–operator gap argument\cite{BourgainGamburd2011}.\footnote{Here, the moment-operator gap refers to the spectral gap of the one-layer second-moment operator
\[
M_2 \;=\; \mathbb{E}\!\left[U^{\otimes 2}\otimes \overline{U}^{\otimes 2}\right],
\]
or equivalently, after local twirling, to the spectral gap of the induced Markov kernel on the reduced commutant basis \(\{I,\mathrm{SWAP}\}^m\); this gap controls the exponential convergence of \(M_2^L\) to the Haar second moment \cite{BrandaoHarrowMixing2016,BourgainGamburd2011}. See App.  \ref{subsec:moment-gap}.}
\end{proof}


\begin{corollary}[Constant-measure good angles under a $2$-design]
\label{cor:pz-constant}
Let $\mathsf U$ be an $\varepsilon$-approximate unitary $2$-design on the encoded space
(e.g., as in Thm.~\ref{thm:global-2design}). For fixed $x^\star$ and input $|\phi\rangle$,
with $X=|\langle x^\star|U|\phi\rangle|^2$ where $U\sim\mathsf U$,
\[
\mathbb E[X]=\tfrac{1}{D}\pm O(\varepsilon),\quad
\mathbb E[X^2]=\tfrac{2}{D(D{+}1)}\pm O(\varepsilon)\footnote{See derivation in App. \ref{sec:detail-second-moment-single-target}},
\]
and Paley--Zygmund gives $\Pr\!\big[X\ge \tfrac{1}{2D}\big]\ge (D{+}1)/(8D)-O(\varepsilon)=\Omega(1)$. Thus,  Paley–Zygmund gives, for sufficiently small \(\varepsilon_{\mathrm{glob}}\),
\begin{equation}
  \Pr\!\Bigl[X \ge \tfrac{1}{2D}\Bigr]
  \;\ge\;
  \frac{\bigl(\mathbb{E}[X]-\tfrac{1}{2D}\bigr)^2}{\mathbb{E}[X^2]}
  \;\gtrsim\; \frac{(1/2D)^2}{2/(D(D{+}1))}
  \;=\; \frac{D{+}1}{8D}
  \;=\; \Omega(1).
  \label{eq:ce-SoverD}
\end{equation}
\end{corollary}

\noindent

\paragraph{Design bounds as controls: Unitary $1$-designs vs Unitary $2$-designs.}
Unitary $1$-designs give \emph{means} while unitary $2$-designs add \emph{dispersion} and \emph{correlation} control and calibrate pairwise covariances of quadratic observables. In combinatorial
optimization this means we can quantify anticoncentration \emph{within} the encoded space; and separate genuine instance–parameter structure (like persistent peaks above $1/D$) from labeling artifacts. In variational training, $2$-design control on second moments stabilizes the statistics of loss estimates (and common gradient/finite-difference estimators that depend on degree-$\le2$ unitary moments), yielding instance–size predictable variance and thus shot requirements. For our constraint-enhanced QAOA, the $1$-design perspective explains \emph{why} favorable parameters must exist (and provides a clean null/control), whereas the $2$-design refinement explains \emph{how reliably} we can \emph{find and certify} them. 

\medskip
\noindent
Finally, a caveat. Unitary $t$–design bounds (including our $1$– and per–block $2$–design controls) are
\emph{conservative}: they certify correct means/second moments under structured randomization, but they do
not translate into tight, instance–specific shot complexity. In particular, the $1$–design baseline $1/D$
is an analytic \emph{null}—useful for existence and for flattening controls via twirls—but it grossly
overestimates the shots needed when constructive interference lifts the success probability well above $1/D$.
Likewise, second–moment control ($O(m\varepsilon)$ product moments) stabilizes estimators and separates signal
from labeling artifacts, yet it is not a substitute for the binomial tails of the \emph{actual} success rate
$p(\gamma,\beta)$ attained by CE–QAOA at good angles. For shot complexity we therefore appeal to the
realized Bernoulli scaling $N=\tilde O(1/p_\star)$ (via Chernoff (Lemma \ref{lem:chernoff})), where in our setting
$p_\star$ empirically behaves like $n^{-k}$ once only $k\ll m$ blocks remain unaligned. This yields polynomial
shot budgets, in stark contrast to the exponential $\Theta(D)$ implicit in the $1$–design null (Cf. Classical competitveness in Sec \ref{sec:classical-competitiveness} vs. Sec \ref{sec:discussion}).

\subsection{Classical Competitiveness I: The Unitary t-Design Baseline}
\label{sec:classical-competitiveness}

The comparison with classical procedures hinges on \emph{what domain the classical sampler can efficiently prepare}. Our CE--QAOA acts \emph{inside} the encoded one–hot product space \(\OH=[n]^m\) of size \(D=n^m\). A classical algorithm, however, may or may not have an efficient generator for \(\OH\). We therefore separate two information models and stay strictly within first/second moment reasoning where permutation twirls achieve overlap at least \(1/D\). Throughout, let \(F\subseteq \OH\) be the feasible set, \(|F|=:S\), and let \(U\) be the depth-\(p\) CE--QAOA circuit restricted to the encoded space.  We now make the “ambient domain” explicit.

\paragraph{Model A (structured domain).}
Both quantum and classical procedures can sample arbitrarily from the encoded domain
\(\OH=[n]^m\) (size \(D=n^m\)), and both have an oracle answering feasibility membership in \(F\).

\paragraph{Model B (raw bitstrings only).}
The classical procedure sees only raw bitstrings of length \(N\) (for TSP/assignment, \(N=n^2\))
and has access to a feasibility predicate \(\mathsf{feas}(z)\in\{0,1\}\) that recognizes whether
a bitstring encodes a feasible element of \(F\subseteq \OH\). No efficient generator for \(\OH\)
is assumed. The quantum procedure, by construction, prepares superpositions \emph{inside} \(\OH\).

\begin{proposition}[Model A: matching \(\Theta(D/S)\) baselines]
\label{prop:modelA}
Suppose both sides can sample uniformly from \(\OH\) and query \(\mathsf{feas}\).
Then any classical algorithm that draws i.i.d.\ proposals from \(\OH\) has success probability \(S/D\) per trial, so the expected number of trials to hit \(F\) is \(D/S\).
For CE--QAOA with a random twirl, Lemma~\ref{lem:perm-twirl} implies the same expected success \(S/D\) for ``hit any feasible'', hence the same \(D/S\) expected trials.
\end{proposition}

\begin{proposition}[Model B: raw-bitstring baseline]
\label{prop:modelB}
Let the classical ambient space be \(\{0,1\}^{N}\) with \(N=n^2\) (for \(n{=}m\)), of size \(2^{N}\),
and assume only the feasibility predicate \(\mathsf{feas}\) is available.
Under the minimax distribution in which \(F\) is uniform among all \(S\)–subsets of \(\{0,1\}^{N}\),
the expected number of feasibility queries until the first hit is
\[
  \mathbb{E}[\#\text{queries}]
  \;=\; \frac{2^{N}+1}{S+1}
  \;=\;\Theta\!\left(\frac{2^{n^2}}{S}\right).
\]
In contrast, CE--QAOA prepares states \emph{inside} \(\OH\)  and,
with only the first–moment guarantee, has expected trials \(D/S=\Theta(n^n/S)\).
Thus the ratio of classical (Model B) to quantum expected trials scales like
\[
  \frac{2^{n^2}/S}{\,n^n/S\,}
  \;=\; \frac{2^{n^2}}{n^n}
  \;=\; \Bigl(\frac{2^n}{n}\Bigr)^{\!n}
  \;=\; \exp\!\bigl(\Theta(n^2)\bigr).
\]
\end{proposition}

\noindent
Thus, CE--QAOA already achieves a conditional quantum-classical separation from design only bounds. However, we as we have already discussed, this is a mere theoretical floor independent of constructive interference present when the twirls are not implemented. As we saw in Fig \ref{fig:full-Blockqaoa}, when the twirling is omitted, heavy outputs appear due to instance dependent constructive interference. As illustrated in Fig \ref{fig:heavy}, these instance dependent heavy outputs pushes the probability of sampling the optimum from the design baseline into the finite probability regime. Additional contributions from parameter optimization are discussed in Sec. \ref{sec:discussion}.

\subsection{Polytime Hybrid Quantum–Classical Solver (PHQC)}
\label{sec:PHQC}

We now extend the algorithmic pipeline to include a \emph{coarse grid} search over  $(\gamma,\beta)$ parameters as used in Sec. \ref{app:1-design} and an \emph{exact classical checker} that returns the lowest energy feasible bitstring regardless of its empirical frequency.  Let \(\beta\in[0,\pi]\), \(\gamma\in[0,\pi]\). For a chosen depth \(p\), we build a rectangular grid
\[
\mathcal G\;=\;\{(\beta_i,\gamma_j):\ \beta_i = i\,\Delta_\beta,\ \gamma_j = j\,\Delta_\gamma,\ 
0\le i\le N_\beta,\ 0\le j\le N_\gamma\},
\]
with spacings \(\Delta_\beta,\Delta_\gamma \in[0,\pi]^{\,n+1}\) as in Theorem \ref{thm:exist-params}. Given a multiset of measured bitstrings \(\mathcal S\), the checker evaluates \emph{exact} objective values \(E(b)=\langle b|H_{\mathrm{obj}}|b\rangle\) for all \(b\in\mathcal S\) and returns the minimal element. Consequently, \emph{it is irrelevant whether the optimal bitstring is the most frequent sample}. Any single appearance of an optimal \(b^\star\) suffices for the checker to output the optimum. This strengthens the overall protocol and minimizes the shot budget. The deterministic post-processing performs scoring \(b\mapsto \langle b|H_C|b\rangle\) in \(O(n^{2})\) with total classical work \(O(S\,n^2)\).

\begin{lemma}[Chernoff bound for optimum hits]\label{lem:chernoff}
Let \(N_{\star}\) be the number of \emph{optimal} bit-strings observed
in \(S\) independent samples, each occurring with probability
\(p_{\min}(\epsilon)\).  Then
\[
    \Pr\!\bigl[N_{\star}=0\bigr]
    \;\le\;
    \exp\!\bigl(-p_{\min}(\epsilon)\,S\bigr),
    \qquad
    S\;\ge\;\Bigl\lceil
        \tfrac{\ln(1/\delta)}{\,p_{\min}(\epsilon)}
    \Bigr\rceil
    \;\Longrightarrow\;
    \Pr[N_{\star}\ge1]\;\ge\;1-\delta .
\]
\end{lemma}

\begin{proof}
Let \(X_{i}\sim\mathrm{Bernoulli}\!\bigl(p_{\min}(\epsilon)\bigr)\) be the
indicator for an optimum hit in shot \(i\).  Then
\(N_{\star}=\sum_{i=1}^{S} X_{i}\) with
\(\E[N_{\star}]=S\,p_{\min}(\epsilon)\).
Using \(\Pr[N_{\star}=0]=(1-p_{\min})^{S}\le e^{-p_{\min}S}\) proves the
first inequality; the shot bound follows by inversion.
\end{proof}

\begin{theorem}[Sample complexity for perfect recovery]
\label{thm:sample_complexity}
Let \(p_{\min}>0\) be the total probability mass on the optimal
bit-strings after \(p\) Constraint-Enhanced QAOA layers.
With
\(
    S\;\ge\;\bigl\lceil\ln(\delta^{-1})/p_{\min}\bigr\rceil
\)
shots, the PHQC post-processor returns a \emph{globally optimal}
solution with probability at least \(1-\delta\).
\end{theorem}
\begin{proof}[Sketch]
Apply Lemma \ref{lem:chernoff} with \(p_{\min}(\epsilon)=p_{\min}\).
\end{proof}

\begin{algorithm}[H]
\caption{\textbf{PHQC} — constant-depth CE--QAOA + deterministic checker}
\label{alg:PHQC1}
\begin{algorithmic}[1]
\Require size $n$; depth $p$; grid $\mathcal G$ over $(\beta,\gamma)$; $H_C=\Hpen+\Hobj$; mixer $H_M$;
         shots $S=\lceil n^{k}\ln(1/\delta)\rceil$.
\Ensure optimal feasible $(b^\star,c^\star)$.
\State $\ket{s_0}\gets \ket{s_{blk}}^{\otimes m}$.
\State $\mathcal S\gets\emptyset$.
\For{each $(\beta,\gamma)\in\mathcal G$}
  \State prepare $\ket{s_0}$; apply $\UC(\gamma)$; apply $\UM(\beta)$; measure $S$ shots into $\mathcal S$.
\EndFor
\State $(b^\star,c^\star)\gets(\text{null},+\infty)$.
\For{each $b\in\mathcal S$}
  \If{\textsc{Feasible}$(b)$}
    \State $c\gets \langle b|\Hobj|b\rangle$ \Comment{or full $H_C$ if you score penalties}
    \If{$c<c^\star$} \State $(b^\star,c^\star)\gets(b,c)$ \EndIf
  \EndIf
\EndFor
\State \Return $(b^\star,c^\star)$.
\end{algorithmic}
\end{algorithm}



\section{Circuit Simulation Results}
\subsection{Traveling Salesman Problem}
\label{subsec:tsp-encoding}

The Traveling Salesman Problem (TSP) is defined on a complete weighted graph
\(G=(V,E)\), \(|V|=n\), with distance matrix \(C_{ab}\in\mathbb{R}_{\ge 0}^{n\times n}\). A tour is a Hamiltonian cycle visiting each city exactly once. We represent a tour by an \(n\times n\) binary assignment matrix \(X=(x_{i,a})\) where
row \(i\) indicates the city visited at \emph{position} \(i\) in the tour and column \(a\)
indicates \emph{which} city. Each position is assigned exactly one city and each city appears exactly once,
\begin{equation}
\sum_{a=1}^{n} x_{i,a}=1\quad(\forall i),\qquad
\sum_{i=1}^{n} x_{i,a}=1\quad(\forall a),\qquad x_{i,a}\in\{0,1\}.
\end{equation}
This is the double one-hot encoding. It matches our block structure defined in Def. \ref{def:kernel-requirement} with each of the \(m=n\) rows defined as a one-hot block of size \(n\). 

\noindent
On the computational basis \(x\mapsto X=(x_{i,a})\)\footnote{We write a computational-basis string $x\in\{0,1\}^{n^{2}}$ in
matrix form $X=(x_{i,a})$ only to make the block structure explicit. One may simply treat $x$ itself as the $n\times n$ assignment matrix with entries $x_{i,a}$. The reshape $x\mapsto X$ merely clarifies the
row/column one–hot constraints and the cyclic pairwise objective
$H_{\mathrm{obj}}(x)=\sum_{i,a,b} C_{ab}\,x_{i,a}x_{i+1,b}$, which naturally
refer to this matrix structure.}
 and the \emph{objective} is the cyclic pairwise cost between consecutive positions:
\begin{equation}
\label{eq:tsp-obj}
H_{\mathrm{obj}}(x)\;=\;\sum_{i=1}^{n}\ \sum_{a,b=1}^{n} C_{ab}\, x_{i,a}\, x_{i{+}1,b}
\qquad (i{+}1\equiv 1).
\end{equation}
Feasibility is enforced by quadratic \emph{penalties} for row/column sums,
\begin{equation}
\label{eq:tsp-pen}
H_{\mathrm{pen}}(x)\;=\;\lambda_{\mathrm{row}}\sum_{i=1}^{n}\!\Bigl(\sum_{a=1}^{n}x_{i,a}-1\Bigr)^{2}
\ +\ \lambda_{\mathrm{col}}\sum_{a=1}^{n}\!\Bigl(\sum_{i=1}^{n}x_{i,a}-1\Bigr)^{2},
\end{equation}
optionally augmented by linear forbids (e.g., disallowing self-loops or precluded edges). The \emph{double one-hot} constraints enforce a permutation of the visited locations. The first term in Eq~\ref{eq:tsp-pen} is now redundant and can be dropped.  The \emph{cost Hamiltonian} is \(H_C=H_{\mathrm{pen}}+H_{\mathrm{obj}}\), which is diagonal in the computational basis and fits our kernel requirements in Def \ref{def:kernel-requirement}. See \cite{Lucas2014Ising} for detailed QUBO derivations and discussions.  

\subsection{Numerical validation on approximate circuit simulators}
\label{subsec:numerics-hw}

We benchmark the CE-QAOA on QOPTLib instances\cite{Osaba2024Qoptlib} by setting a fixed start city (``anchored'' reduction), which yields $(n{-}1)^2$ logical qubits (blocks $=n{-}1$, block size $=n{-}1$) for a problem defined on $n$ locations. We work with a single layer composed as $U(\gamma,\beta)=U_M(\beta)\,U_C(\gamma)$. We sweep a grid $\mathcal G \subset [0,\pi]^2$ with $|\mathcal G|=(n{+}1)^2$ angle pairs, and use a shot budget of \(S\;=\;10\,n^{k}\) with $k$ ranging from $3-5$ which comfortably absorbs finite-precision and rounding effects of the grid search  and finite sampling since the number of qubits scales as $O(n^2)$. For each grid point we measure in the computational basis, classically filter feasibile solutions, and evaluate cyclic tour costs relative to the fixed start. We report (i) the best measured feasible solution, (ii) the fraction of feasible outcomes, and (iii) the minimum tour cost. In Table \ref{tab:design-bound}, the PHQC recovered the optimum in all cases. The fact that Algorithm \ref{alg:PHQC1} guarantees any sampled optimal bitstring to be identified regardless of its frequency, greatly reduces the shot burden in our implementation. \footnote{For the larger system sizes (dj8-dj10) the per grid shot budget had to be raised because the circuit approximation struggled to capture the full dynamics. Similarly, for our smaller instances where full circuit simulation was possible, a lower $10n^2$ shots was found sufficient without increasing layer or grid size. }

\begin{table}[H]
\centering
\caption{Single–layer CE--QAOA sampling with design–based angles and shot budget \(S=10\,n^{3-5}\) used to recover the global optimum. For a fair comparison, we assume a classical sampler with access to the block one-hot space so that the classical shot cost formula \(S_{\mathrm{cl}} \;=\; \Theta\!\big(n^{\,m-1}\log(1/\delta)\big)
\)  holds. We set $\log(1/\delta)=10$ for both quantum and classical models and $n^r$ refers to the ratio in Eq. \ref{eq:adva}.  } 
\label{tab:design-bound}

\adjustbox{max width=\linewidth}{%
\begin{tabular}{@{}lrrrrll@{}}
\toprule
Instance & Cities \(n\) & Shots \(S\) & Advantage $(n^{\,r})$ & Tour cost & Optimal & $(\gamma,\beta)$ \\
\midrule
\texttt{wi4}  & 4  & \(160\)              & $n^1$    & 6{,}700 & Yes      & \((1.57,2.36)\) \\
\texttt{wi5}  & 5  & \(250\)          & $n^2$   & 6{,}786 & Yes      & \((1.57,2.36)\) \\
\texttt{wi6}  & 6  & \(360\)          & $n^3$   & 9{,}815 & Yes      & \((2.62,2.09)\) \\
\texttt{wi7}  & 7  & \(733\)          & $>n^3$   & 7{,}245 & Yes      & \((3.14, 2.62)\) \\
\texttt{dj8}  & 8  & \(2.56{\times}10^4\) & $>n^3$& 2{,}762 & Yes      & \((1.35,2.24)\) \\
\texttt{dj9}  & 9  & \(3.65{\times}10^4\) & $>n^4$& 2{,}134 & Yes      & \((0.30,1.60)\) \\
\texttt{dj10} & 10 & \(1.00{\times}10^6\)  & $ > n^4$& 2{,}822 & Yes      & \((1.01, 1.57)\)
\\
\bottomrule
\end{tabular}}
\end{table}

 \footnote{A note on classically sampling from feasible permutations.
Even if a classical sampler draws \emph{directly} and uniformly from feasible bitstrings, the success probability for a unique optimum is
\[
p_{\mathrm{succ}}=\frac{1}{n!}
\quad\Rightarrow\quad
S_{\mathrm{cl}}=\Theta\!\big(n!\,\log(1/\delta)\big).
\]
Given the rapid growth of \(n!\), the advantage persists, with tweaks in prefactors and orders of magnitude.
}

\subsection{Classical Competitiveness II: \texorpdfstring{$\Theta(n^{r})$}{Theta(n\string^r)} Sampling Advantage}

\label{sec:discussion}
The $\varepsilon$-approximate $2$-design on the \emph{encoded} manifold gives Haar-like second-moment anticoncentration at scale $1/D$ with $D=n^{m}$ (Thm.~\ref{thm:global-2design}), which is a \emph{conservative} worst-case sampling floor and a predictor of typical-case overlap statistics (Cor.~\ref{cor:pz-constant}). However, there are two empirical observations to be made at depth $p{=}1$. (i) Polynomial-shot recovery across instances with a basic coarse grid search over the $(\gamma,\beta)$ parameter space. (ii) Systematic \emph{peaks above $1/D$} when the block-permutation twirl is \emph{omitted} (Fig.~\ref{fig:grid-search}) suggesting that the native labeling preserves instance structure and constructive interference. This is expected to remain the case at higher layers with $p>1$. These \emph{empirical} observations therefore \emph{motivate} a targeted search over angle pairs and shallow depths to find constructive-interference settings that \emph{achieve} per-shot success at the heuristic scale
\begin{equation}
\label{eq:cascade}
p_{\min}\ \gtrsim\ \frac{1}{n^{k}}\quad\text{for }k < m.\ 
\end{equation}

\medskip
\noindent
To see how this might arise under finer angle  and layer optimization, we can think of each fixed block in $m$ as systematically reducing the effective degrees of freedom from $m$ blocks down to $k$. We start in a uniform superposition of block one-hot basis vectors and assume that at least one block can be fixed at the beginning of the optimization process. This is no different from fixing the starting position in the TSP. The parameter optimization process fixes further block(s) to minimize the objective function. Thanks to the global constraint, each layer can be further assumed to slightly favor configurations that agree with a target pattern on a few blocks with a suitable angle choice. These blocks are then assumed \emph{locked}. Meaning that finer parameter optimization at current layer or next can yield a new set of parameters that improves the number of locks but not degrade it. We can think of a block as \emph{locked} once the target symbol in that block gains a fixed margin over the uniform $1/n$ baseline and $k$ is interpreted  as residual ``undecided'' blocks.  Because of the global constraint, subsequent layers leverage already–locked blocks to bias additional ones, producing a cascade that leaves only $k$ blocks unfixed.

\medskip
\noindent
At that point the success probability depends only on these $k$ degrees of freedom and is therefore bounded
below by $1/n^{\,k}$ independent of the ambient dimension $D=n^m$. The success probability in Eq. \ref{eq:cascade} follows. PHQC \emph{then} converts this into \emph{polynomial} shot complexity,
\(
S\ \ge\ \lceil \log(1/\delta)/p_{\min}\rceil
= O\!\big(n^{k}\log(1/\delta)\big)
\)
(Thm.~\ref{thm:sample_complexity}). With the classical shot cost given as  \(S_{\mathrm{cl}} \;=\; \Theta\!\big(n^{\,m-1}\log(1/\delta)\big)
\) shots, the ratio is

\begin{equation}
\label{eq:adva}
\frac{S_{\mathrm{cl}}}{S_{\mathrm{q}}} \;=\; \Theta\!\big(n^{\,r}\big),
\end{equation}
Where \(r=m-1-k\). Taking into account the clear upward trend in Table \ref{tab:design-bound}; with  $r = 6$ for $n=15$ problem,  the quantitative advantage exceeds $10^7$. It goes without saying that better circuit approximations would improve accuracy and amplify the advantage.

\subsubsection{Near-Optimality of Grid Parameters}
\noindent
It is also worth noting that our coarse parameter grid is far from optimal. We use it here as a sufficient condition for the existence of parameters that can yield the design based results and subsequent large reduction in shot cost expected when the twirling averaging is bypassed. The approximate 2-design results suggest inherent trainability and anticoncentration in the proposed construction. Empirically, increasing the the grid points yields more frequent near-optimal interference within the same shot budget. For example, at $n=8$, using $20\times 20$ grids for $(\gamma,\beta)$ (instead of $(n{+}1)\times(n{+}1)=9\times 9$) and the same shot budget of $10n^3=5120$ produced $12$ angle pairs attaining the optimum, and $40$ angle pairs within $1\%$ of the optimum and $229$ points within $5\%$ of the optimum. Whereas the coarser grid of size $n+1$ found only $2$ pairs of grid points attaining the optimum. Similar phenomenon was observed for problems with $n = 9 $ and $10$ respectively. This supports our choice to prioritize \emph{density over range} and treat $[0,\pi]^2$ as a sufficient compact tile for practical scans, while leaving systematic parameter optimization for future work.

\paragraph{Approximate circuit simulation as probes.}
All circuit evaluations are performed using a \emph{matrix–product–state} (MPS) simulator in the sense of Vidal’s tensor–network formulation~\cite{Vidal2003Efficient}. Specifically, we used the \texttt{matrix\_product\_state} backend of the IBM Qiskit \textsc{Aer} simulator\cite{Qiskit2023}. For reproducibility, we configured the backend in a deliberately \emph{aggressive truncation} regime with maximum bond dimension of $128$, truncation and validation thresholds of $10^{-3}$, and an amplitude--chopping threshold of $10^{-3}$. Even under these approximations, the CE--QAOA circuit consistently recovers optimal tours on instances of size $4$ to $10$ nodes. This strongly suggests that the diagonal entangler (Prop.~\ref{prop:diag-entangler}) and constant-gap block–XY mixer (Prop.~\ref{prop:spectral-gap}) induce useful interference patterns already at modest entanglement, aligning with the encoded $2$-design baseline (Thm.~\ref{thm:global-2design}) as a conservative control. Because truncation damps constructive interference, we expect the non-twirled peaks to sharpen under higher-fidelity simulation, increasing the single-shot success $p_{\min}$ and directly reducing the required shots at fixed confidence. Larger instances ($15$–$25$ nodes) were probed only at zero angles due to the prohibitive wall–time for full MPS simulation on commodity hardware. 

\medskip
\noindent
The truncations we employ motivate near-term \emph{hardware probes} on devices that admit the circuit width/depth of these instances or their compressed variants. Our grid over $(\gamma,\beta)$, define a calibration surface for \emph{algorithm--hardware co-design}. Compilers can target mixer locality; error-mitigation can prioritize preserving diagonal-phase coherence; and hardware can be tuned to maintain the short-depth interference that creates the observed peaks. The encoded $2$-design baseline (Thm.~\ref{thm:global-2design}) supplies a Haar-like \emph{null model} at scale $1/D$ against which hardware-induced deviations can be quantitatively assessed. We therefore propose our CE-QAOA kernel with Tables \ref{tab:design-bound} as references for \emph{probing} future hardwares. 

\paragraph{PHQC vs.\ classical exact/heuristic solvers.}
High-quality heuristics (like LKH implementation) routinely find near-optimal tours quickly~\cite{LinKernighan1973}, and exact methods (e.g., branch-and-cut as in \textsc{Concorde}) provide provably optimal solutions in practice on large benchmarks~\cite{Applegate2007}. Dynamic programming à la Bellman–Held–Karp offers a systematic exact approach with worst-case complexity $T_{\mathrm{HK}}=\Theta(n^{2}2^{n})$ time and $M_{\mathrm{HK}}=\Theta(n2^{n})$ memory~\cite{Bellman1962}. By contrast, our PHQC protocol gives a \emph{probabilistic} recovery guarantee where a per-shot optimal-hit probability $p_{\min}$ at selected shallow angles, $S\ge \lceil \ln(1/\delta)/p_{\min}\rceil$ shots suffice to obtain the optimal tour with probability at least $1-\delta$ (Thm.~\ref{thm:sample_complexity}). This is incomparable to the worst-case exponential bound of dynamic programming. While PHQC provides randomized \emph{anytime} certificates (via repeats to boost $1-\delta$), DP provides deterministic exactness with exponential resources. We leave a rigorous head-to-head comparison between exact solvers and PHQC to future work. However, we collect the first indication of the ability of the instance dependent heavy outputs to keep the per-shot success probability in the finite shot region in the Table \ref{tab:emp-verify} and Fig. \ref{fig:heavy-outputs}.

\begin{figure}[t]
  \centering
  \includegraphics[width=0.64\linewidth]{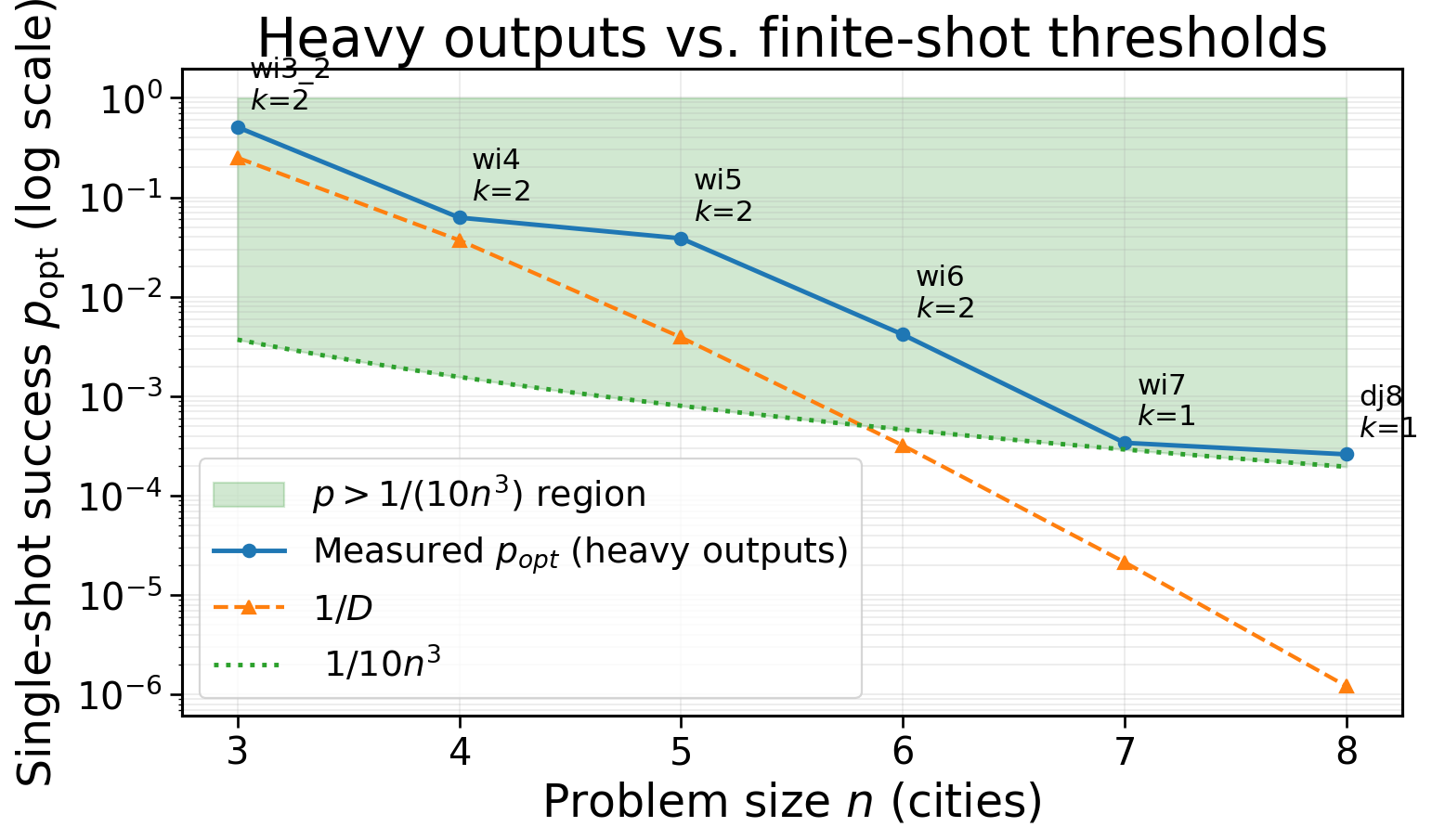}
  \caption{\textbf{Heavy outputs vs.\ finite-shot thresholds.}
    Shaded green shows the finite-shot region \(p \ge \tfrac{1}{10n^3}\).
    Theorem-level guarantees (yellow boundary in schematic figures) place us above this curve,
    while instance-dependent heavy outputs (blue points/curve) push into the finite-shot regime Cf. \ref{fig:heavy}.}
  \label{fig:heavy-outputs}
\end{figure}

\begin{table}[t]
  \centering
  \footnotesize
  \setlength{\tabcolsep}{3pt}
  \renewcommand{\arraystretch}{1.08}
  \begin{tabular}{@{}lrrrrlrlr@{}}
    \toprule
    Inst. & $n$ & $n_{\mathrm{red}}$ & \multicolumn{1}{r}{Shots} &
    \multicolumn{1}{r}{$k$} & $(\gamma,\beta)$ &
    \multicolumn{1}{r}{Opt} & Ver. & \multicolumn{1}{r}{$p_{\mathrm{opt}}$} \\
    \midrule
    \texttt{wi3\_2} & 3 & 2 & 540  & 2 & (0.57, 0.36) &
    5456 & PASS & \(5.1\times10^{-1}\) \\

    \texttt{wi4}    & 4 & 3 & 960  & 2 & (1.57, 2.36) &
    6700 & PASS & \(6.3\times10^{-2}\) \\

    \texttt{wi5}    & 5 & 4 & 1500 & 2 & (2.20, 2.44) &
    6786 & PASS & \(3.9\times10^{-2}\) \\

    \texttt{wi6}    & 6 & 5 & 2160 & 2 & (1.01, 1.75) &
    9815 & PASS & \(4.2\times10^{-3}\) \\

    \texttt{wi7}    & 7 & 6 & 2940 & 1 & (1.35, 2.69) &
    7245 & PASS & \(3.4\times10^{-4}\) \\

    \texttt{dj8}    & 8 & 7 & 3840 & 1 & (3.14, 1.35) &
    2762 & PASS & \(2.6\times10^{-4}\) \\
    \bottomrule
  \end{tabular}
  \caption{Empirical verification of polynomial single shot success probability from non-twirled circuit outputs. $k$ is the degeneracy observed in the sampled bitstrings while $n_{red}$ is teh number of blocks left after fixing the starting position. The twirling baseline \(=1/D\). The choice of \(\ln(1/\delta)=10\) corresponds to \(\delta=e^{-10}\approx4.5\times10^{-5}\), a very stringent failure probability (\(\approx 0.005\%\)). \(p_{\mathrm{opt}}\) sums probabilities over all degenerate optimal bitstrings. }
  \label{tab:emp-verify}
\end{table}

\section{Conclusion}
\label{sec:conclusion}

\medskip \noindent
We introduced CE--QAOA, a constraint–aware kernel that (i) prepares and preserves a block one–hot manifold with an ancilla–free, depth–optimal encoder and a constant–gap two–local block–XY mixer; (ii) establishes low–order unitary moments—an exact $1$–design via block–permutation twirling and an approximate $2$–design from short per–block XY evolutions combined with a diagonal inter–block entangler; (iii) yields a Paley–Zygmund anticoncentration and typical–case overlap guarantees at the $1/D$ scale for encoded dimension $D=n^m$; (iv) couples naturally to a deterministic classical post–processor (PHQC) that converts overlap into an end–to–end solver with explicit shot complexity; and (v) frames a quantum–classical separation based on access models that make explicit what the classical sampler can generate efficiently. Algorithmically, PHQC identifies the best feasible sample in $O(Sn^2)$ classical time. This implies that obtaining a shot complexity independent of the encoded dimension $D$ is the path to an unconditional quantum advantage with a polynomial time quantum-classical solver. We show empirically that in absence of permutation twirling, instance dependent interference leads to heavy ouputs in the probability profile far above the design baseline into the polynomial shot regime.

\medskip \noindent
 We subsequently develop an inductive explanation to suggest that this arises in CE-QAOA when only a fraction of blocks remain unaligned on the target solution during parameter optimization. Our observation suggests that when CE–QAOA is seeded with a fixed starting block, it concentrates weight so that success depends on the number $k$ of \emph{unaligned} blocks at the end of state preparation. Typically $k\ll m$, and the effective scale shifts from $1/n^{m}$ to about $1/n^{k}$, yielding orders–of–magnitude gains. Thus, we obtain a \(\Theta(n^r)\) shot–complexity advantage whenever CE–QAOA locks \(r\!\ge\!1\) blocks. This advantage would persist even if a classical sampler drew uniformly from the feasible space itself because classical sampling offers no mechanism to overcome the exponential suppression of success probability. We leave a rigorous proof of this phenomenon to future work.

\medskip \noindent
However, when a classical competitor only samples raw $N=n^2$ bitstrings and relies on a feasibility predicate, while CE--QAOA prepares states directly in $\OH$, the expected trial counts differ by an $\exp[\Theta(n^2)]$ factor in a minimax sense. We view these results as a step toward symmetry–inspired, shallow quantum heuristics with analyzable typical–case behavior and clean interfaces to efficient classical checking. Finally, our paper illustrates the practical impact of encoded designs in optimization problems. We observe that due to the typical behaviour results that emerge from this method, the bounds always tend to be conservative. We hope that future works are able to go beyond unitary designs to prove a worst case shot complexity independent of the encoded dimension $D$, paving the way for unconditional quantum advantage.



\begin{appendix}

\section{CE-QAOA Circuit Diagrams and Block-Constraint Example}
\label{app:blocks}

In a typical QAOA formulation with $N$ qubits, the full Hilbert space is 
\[
\mathcal{H}^{\otimes N} \quad \text{with dimension} \quad 2^N.
\]
For the case where $N=n^2$ (i.e., $n^2$ qubits), the full space has dimension 
\[
\dim(\mathcal{H}_{\text{full}}) = 2^{n^2}.
\]

Suppose we partition the $n^2$ qubits into $n$ blocks of $n$ qubits each. We require that each block is restricted to states of the form
\[
|0\cdots 1_i \cdots 0\rangle,\quad i=0,\ldots,n-1.
\]
The projector onto the one-hot subspace for block $j$ is defined as
\[
P_n^{(j)} = \sum_{i=0}^{n-1} \bigl|0\cdots 1_i \cdots 0\bigr\rangle_j\bigl\langle 0\cdots 1_i \cdots 0\bigr|_j.
\]
The overall projector on the full system is then
\[
P = \bigotimes_{j=0}^{n-1} P_n^{(j)},
\]
which restricts the Hilbert space to the tensor product of one-hot subspaces. Since each block now has dimension $n$, the dimension of the one-hot subspace is
\[
\dim(\mathcal{H}_{\text{one-hot}}) = n^n.
\]




Here we illustrate the Kernel structure of Definition \ref{def:kernel-requirement} for the Traveling Salesman Problem\cite{Lucas2014Ising}. Suppose we have a \(3\)-city TSP-like problem requiring a \(3\times 3\) binary matrix encoding. In the block-wise viewpoint we need 9 qubits qrouped into 3 blocks, representing the roes of the permutation matrix (each block represents which city is in the 1\textsuperscript{st}, 2\textsuperscript{nd}, or 3\textsuperscript{rd} position). Each block has 3 qubits (exactly one of the 3 qubits is \(\ket{1}\), the others \(\ket{0}\)). A valid bitstring in the computational basis then takes the form
\[
\bigl(\,b_0^0\;b_0^1\;b_0^2\bigr)\;\bigl|\,\bigr.\;\bigl(\,b_1^0\;b_1^1\;b_1^2\bigr)\;\bigl|\,\bigr.\;\bigl(\,b_2^0\;b_2^1\;b_2^2\bigr),
\]
where each parenthetical group has exactly one \(\ket{1}\). For example:
\begin{align*}
&\underbrace{100}_{\text{Block 0}}\;\underbrace{010}_{\text{Block 1}}\;\underbrace{001}_{\text{Block 2}},\\
&\underbrace{010}_{\text{Block 0}}\;\underbrace{100}_{\text{Block 1}}\;\underbrace{001}_{\text{Block 2}},\\
&\underbrace{001}_{\text{Block 0}}\;\underbrace{100}_{\text{Block 1}}\;\underbrace{010}_{\text{Block 2}},
\end{align*}
etc. Each triple of bits has exactly one \(\ket{1}\) and so on.

Suppose we have a \(4\)-city TSP-like problem requiring a \(4\times 4\) binary matrix encoding. The block structure in the computational basis then takes the form
\[
\bigl(b_0^0\;b_0^1\;b_0^2\;b_0^3\bigr)\;\bigl|\;\bigl(b_1^0\;b_1^1\;b_1^2\;b_1^3\bigr)\;\bigl|\;\bigl(b_2^0\;b_2^1\;b_2^2\;b_2^3\bigr)\;\bigl|\;\bigl(b_3^0\;b_3^1\;b_3^2\;b_3^3\bigr),
\]
where each parenthetical group has exactly one \(\ket{1}\). It is worth emphasizing again that not all the bitstrings satisfying this blockwise constraint are valid tours, however, all valid tours satisfy this structure. So for a TSP, the constraint is only partially encoded.

\subsection{Permutation Group Actions}
\label{app:perm}


\medskip
\noindent
\textbf{Global feasibility–preserving action}.
Let
\(
\mathcal A \;\equiv\; S_m\times S_n \;\subset\; U(\OH),
\)
act by
\begin{equation}
\label{eq:A-action}
\mathsf P_{\sigma,\tau}\,
\bigl(\ket{e_{j_0}}\!\otimes\!\cdots\!\otimes\!\ket{e_{j_{m-1}}}\bigr)
\;=\;
\ket{e_{\tau(j_{\sigma^{-1}(0)})}}\!\otimes\!\cdots\!\otimes\!\ket{e_{\tau(j_{\sigma^{-1}(m-1)})}},
\end{equation}
i.e., a row permutation $\sigma\in S_m$ and the \emph{same} column relabeling
$\tau\in S_n$ applied to every block. When $H_{\mathrm{pen}}$ is
$S_m\times S_n$–invariant (Def.~\ref{def:kernel-requirement}), the feasible set
\[
\mathcal X \;:=\; L_0(H_{\mathrm{pen}})=\{\,x:\langle x|H_{\mathrm{pen}}|x\rangle=0\,\}
\]
is preserved setwise by $\mathcal A$.

\medskip
\noindent
\textbf{Block Level feasibility–preserving action}.

Let
\[
\mathcal B \;\equiv\; S_n^{\times m} \;=\; \Bigl\{\;\mathsf P=\bigotimes_{b=0}^{m-1} \mathsf P_b \;:\; \mathsf P_b\in S_n\;\Bigr\}
\;\subset\; U(\OH),
\]
which applies \emph{independent} label permutations in each block. This action
is \emph{not} feasibility–preserving in general (it can break column one–hot), but it
does implement a first–moment (\(1\)-design) average on the encoded space. By a \emph{permutation twirl} we mean the feasibility–preserving group action
\[
\mathcal A \;=\; S_m^{\text{(rows)}} \times S_n^{\text{(global cols)}},
\]
which (i) permutes the \emph{rows/blocks} uniformly via a single $\pi\in S_m$ and (ii) applies the \emph{same} column relabeling $\tau\in S_n$ to \emph{every} row. This action preserves the set of permutation matrices (feasible states). 

\medskip
\noindent
Recall. If we encode an \(m\times n\) assignment (rows \(b\in[m]\), columns \(j\in[n]\)) by an \(m\)-block one–hot pattern with \(n\) qubits per block. A feasible bitstring is
a \emph{permutation matrix} \(X\in\{0,1\}^{m\times n}\) with exactly one ``1'' in each row and each column. Let \(S_m\) act on rows (blocks) and \(S_n\) act on columns (symbols) with action on \(X\) given by
\begin{align}
(R_\pi X)_{b,j} &= X_{\pi^{-1}(b),\,j}, \qquad \pi\in S_m,
\\
(C_\tau X)_{b,j} &= X_{b,\,\tau^{-1}(j)}, \qquad \tau\in S_n,
\end{align}
and combined action \(X\mapsto R_\pi C_\tau X\).

\begin{proposition}[Global Feasibility is preserved by \(\mathcal A\)]
\label{prop:feasibility-preserved}
If \(X\) is a permutation matrix and \((\pi,\tau)\in\mathcal A\), then \(X' := R_\pi C_\tau X\) is also a permutation matrix.
\end{proposition}

\begin{proof}
Any feasible \(X\) encodes a permutation \(\sigma:[m]\to[n]\) via
\(X_{b,j}=1 \iff j=\sigma(b)\). Under the group action,
\[
(R_\pi C_\tau X)_{b,j}=1
\iff
X_{\pi^{-1}(b),\,\tau^{-1}(j)}=1
\iff
\tau^{-1}(j)=\sigma(\pi^{-1}(b))
\iff
j=\tau\!\bigl(\sigma(\pi^{-1}(b))\bigr).
\]
Thus \(R_\pi C_\tau X\) encodes the permutation \(\tau\circ\sigma\circ\pi^{-1}\).
It remains one-hot in each row and each column, i.e.\ feasible.
\end{proof}

For a qubit-level description, label qubits as \((b,j)\) with \(b\in[m]\) (block/row) and \(j\in[n]\) (column). Then
\begin{align}
(b,j) &\xmapsto{R_\pi} (\pi(b),\,j),
\\
(b,j) &\xmapsto{C_\tau} (b,\,\tau(j)),
\end{align}
and the unitary implementing \(R_\pi C_\tau\) is a permutation matrix over the full
computational basis. Because \(H_{\mathrm{pen}}\) is invariant under \(\mathcal A\), its zero-level set  is stable under these unitaries. This is
the precise sense in which ``the permutation preserves feasibility'' in our kernel.
\paragraph{ \(3\times 3\) (9-qubit) example}
Take \(m=n=3\), so a feasible state is a \(3\times 3\) permutation matrix. Start from
\(\sigma=\mathrm{id}\),
\[
X \;=\;
\begin{pmatrix}
1&0&0\\
0&1&0\\
0&0&1
\end{pmatrix}.
\]
Apply a row cycle \(\pi=(1\,2\,3)\): \(R_\pi X\) is
\(
\begin{psmallmatrix}
0&1&0\\
0&0&1\\
1&0&0
\end{psmallmatrix}
\),
still feasible.
Apply a \emph{global} column swap \(\tau=(1\,3)\) to all rows:
\(
C_\tau X=
\begin{psmallmatrix}
0&0&1\\
0&1&0\\
1&0&0
\end{psmallmatrix}
\),
still feasible.
The combined action \(R_\pi C_\tau X\) remains a permutation matrix, encoding
\(\tau\circ\sigma\circ\pi^{-1}\).

\subsection{A Quick Schmidt–Rank Primer}

\begin{definition}[Schmidt decomposition and Schmidt rank]
Let $\ket{\psi}\in\mathcal H_L\otimes\mathcal H_R$ be a pure state. There exist orthonormal sets
$\{ \ket{u_k} \}_{k=1}^r \subset \mathcal H_L$, $\{ \ket{v_k} \}_{k=1}^r \subset \mathcal H_R$,
and strictly positive numbers $\lambda_1,\dots,\lambda_r$ such that
\[
  \ket{\psi} \;=\; \sum_{k=1}^{r} \lambda_k\, \ket{u_k}_L \otimes \ket{v_k}_R.
\]
The integer $r$ is the \emph{Schmidt rank}, denoted $\mathrm{SR}(\ket{\psi})$ (with respect to the bipartition $L|R$).
\end{definition}

\paragraph{Matrix-reshaping.}
Fix product bases on $L$ and $R$. If one reshapes the amplitude tensor of $\ket{\psi}$ into a matrix
$M \in \mathbb C^{(\dim \mathcal H_L)\times (\dim \mathcal H_R)}$ by grouping all ``left'' indices as rows
and all ``right'' indices as columns, then
\[
  \mathrm{SR}(\ket{\psi}) \;=\; \operatorname{rank}(M).
\]
Thus, the Schmidt rank of a state equals the usual linear-algebra rank of its reshaped coefficient matrix.

\paragraph{Example: $\ket{0}^{\otimes n}$ has Schmidt rank $1$ across any single-edge cut.}
For any cut $\{0,\dots,i-1\}\,|\,\{i,\dots,n-1\}$ one has the factorization
\[
  \ket{0}^{\otimes n} \;=\; \Bigl(\ket{0}^{\otimes i}\Bigr) \otimes \Bigl(\ket{0}^{\otimes (n-i)}\Bigr),
\]
i.e.\ a \emph{single} product term. Hence $\mathrm{SR}=1$. In the matrix view, the reshaped matrix has one nonzero entry
(at the all-zero row/column), so its rank is $1$.

\paragraph{Example: $\ket{W_n}$ has Schmidt rank $2$ across any single-edge cut.}
Recall
\[
  \ket{W_n} \;=\; \frac1{\sqrt n}\sum_{j=0}^{n-1} \ket{0\cdots 1_j\cdots 0}.
\]
Across the cut $\{0,\dots,i-1\}\,|\,\{i,\dots,n\}$ split the excitation ``on the left'' vs.\ ``on the right'':
\[
  \ket{W_n} \;=\; \sqrt{\frac{i}{n}}\;\ket{L_1}\otimes\ket{0\cdots 0}_R
  \;+\;
  \sqrt{\frac{n-i}{n}}\;\ket{0\cdots 0}_L\otimes\ket{R_1},
\]
\[
  \ket{L_1}=\frac{1}{\sqrt i}\sum_{j=0}^{i-1} \ket{0\cdots 1_j\cdots 0}_L,
  \qquad
  \ket{R_1}=\frac{1}{\sqrt{n-i}}\sum_{j=i}^{n-1} \ket{0\cdots 1_j\cdots 0}_R
\]
are orthonormal to the respective all-zero vectors. The two product terms are orthogonal and both coefficients
are nonzero for $1\le i\le n-1$, hence $\mathrm{SR}(\ket{W_n})=2$. In the matrix view, the reshaped matrix has two
nonzero singular values $\sqrt{i/n}$ and $\sqrt{(n-i)/n}$.

\subsection{Two–local and number–conserving structure}
\label{app:two-local}
On a given block $b$ with sites $i\in\{0,\dots,n-1\}$, define the Pauli ladder operators
\[
  \sigma_i^{+} \;=\; \ket{1}\!\bra{0}_i \;=\; \tfrac{1}{2}\bigl(X_i - i Y_i\bigr),
  \qquad
  \sigma_i^{-} \;=\; \ket{0}\!\bra{1}_i \;=\; \tfrac{1}{2}\bigl(X_i + i Y_i\bigr).
\]
A straightforward calculation gives
\begin{equation}
\label{eq:XY-sigmapm}
  X_i X_j + Y_i Y_j
  \;=\;
  2\bigl(\sigma_i^{+}\sigma_j^{-} + \sigma_i^{-}\sigma_j^{+}\bigr),
  \qquad (i\neq j).
\end{equation}
Hence the block mixer can be written (up to an overall factor $2$ absorbable into the angle)
\begin{align}
\label{eq:HM-block}
  H_M^{(b)}
  \;:=\; \sum_{0\le i<j\le n-1} \bigl(X_{i}^{(b)}X_{j}^{(b)}+Y_{i}^{(b)}Y_{j}^{(b)}\bigr)
  \;=\; 2 \sum_{0\le i<j\le n-1}\bigl(\sigma_{bi}^{+}\sigma_{bj}^{-}
                         + \sigma_{bi}^{-}\sigma_{bj}^{+}\bigr).
\end{align}

\paragraph{Action on the one–excitation basis.}
Let $\mathcal H_1=\mathrm{span}\{\ket{e_0},\dots,\ket{e_{n-1}}\}$ denote the one–hot subspace on the block,
where $\ket{e_k}$ has a single `1' at position $k$. For $i\neq j$,
\[
  \sigma_j^{+}\ket{e_i}=\ket{\dots 1_j,1_i,\dots},\qquad
  \sigma_i^{-}\sigma_j^{+}\ket{e_i}=\ket{e_j},\qquad
  \sigma_i^{+}\sigma_j^{-}\ket{e_i}=0.
\]
Using \eqref{eq:XY-sigmapm},
\begin{equation}
\label{eq:XY-hop}
  (X_iX_j+Y_iY_j)\ket{e_i}
  \;=\; 2\,\sigma_i^{-}\sigma_j^{+}\ket{e_i}
  \;=\; 2\ket{e_j}.
\end{equation}
By symmetry,
\(
  (X_iX_j+Y_iY_j)\ket{e_j}=2\ket{e_i},
\)
and for any $k\notin\{i,j\}$,
\(
  (X_iX_j+Y_iY_j)\ket{e_k}=0.
\)
Thus, \emph{restricted to $\mathcal H_1$}, $H_M^{(b)}$ acts as the (weighted) adjacency of the complete graph $K_n$:
\begin{equation}
\label{eq:AdjK}
  \bigl.H_M^{(b)}\bigr|_{\mathcal H_1}
  \;=\; 2\sum_{0\le i<j\le n-1}\bigl(\ket{e_i}\!\bra{e_j}+\ket{e_j}\!\bra{e_i}\bigr)
  \;=\; 2\,A(K_n).
\end{equation}

Excitation number preservation is straight forward. Let $\hat N=\sum_{i=0}^{n-1} \sigma_i^{+}\sigma_i^{-}$ be the excitation–number operator on the block.
Each term $\sigma_i^{+}\sigma_j^{-}$ lowers at $j$ and raises at $i$, keeping the total count unchanged; likewise for its Hermitian conjugate. Hence
\[
  [\,H_M^{(b)},\,\hat N\,]=0,
\]
so $H_M^{(b)}$ preserves all fixed–Hamming–weight sectors, in particular $\mathcal H_1$.

\subsection{Proof of Proposition \ref{prop:controllability-universality}}
\label{app:proof}
\begin{proof}
\textbf{Step 1 (Representation on $\mathcal H_1$).}
With basis $\{\ket{e_k}\}_{k=0}^{n-1}$, the XY terms act as
\[
H_{ij}^{(b)}\big|_{\mathcal H_1} \;=\; E_{ij}+E_{ji},\qquad E_{ij}:=\ket{e_i}\!\bra{e_j},
\]
and $H_Z^{(1)}\big|_{\mathcal H_1}=D:=\sum_{k=0}^{n-1} d_k E_{kk}$ is diagonal. “Nontrivial” means $D\not\propto I$. 

\textbf{Step 2 (Skew-symmetric off-diagonals).}
For any pair $(i,j)$ with $d_i\neq d_j$,
\[
\bigl[D,\; E_{ij}+E_{ji}\bigr] \;=\; (d_i-d_j)\,(E_{ij}-E_{ji}),
\]
so $E_{ij}-E_{ji}$ lies in the Lie closure $\mathfrak L := \mathrm{Lie}\{\, i\mathcal G\}$.

\textbf{Step 3 (Traceless diagonals).}
With $S_{ij}:=E_{ij}+E_{ji}$ and $A_{ij}:=E_{ij}-E_{ji}$ obtained as above for some pair $(i,j)$ with $d_i\neq d_j$, compute
\[
\bigl[S_{ij},\, A_{ij}\bigr] \;=\; 2\,(E_{ii}-E_{jj}) \;\in\; \mathfrak L.
\]
Hence $\Delta_{ij}:=E_{ii}-E_{jj}\in\mathfrak L$ for at least one $(i,j)$. Since commutators of such $\Delta$’s generate the linear span of all traceless diagonals, we can obtain
\[
\mathrm{span}\{\,E_{ii}-E_{jj}\,:\, i\neq j\} \;\subseteq\; \mathfrak L .
\]

\textbf{Step 4 (All skew-symmetric off-diagonals).}
Given any $k\neq \ell$, use the available XY generator $S_{k\ell}=E_{k\ell}+E_{\ell k}\in\mathfrak L$ and a traceless diagonal that separates $k$ and $\ell$, namely $\Delta_{k\ell}=E_{kk}-E_{\ell\ell}\in\mathfrak L$ from Step~3, to get
\[
\bigl[\Delta_{k\ell},\, S_{k\ell}\bigr] \;=\; 2\,(E_{k\ell}-E_{\ell k}) \;\in\; \mathfrak L.
\]
Thus for every $k\neq \ell$ we have both $S_{k\ell}$ and $A_{k\ell}$.

\textbf{Step 5 (Lie algebra identification).}
The set
\[
\bigl\{\, i(E_{k\ell}+E_{\ell k}),\ i(E_{k\ell}-E_{\ell k}) \ (k<\ell);\ i(E_{kk}-E_{\ell\ell})\ (k<\ell)\,\bigr\}
\]
spans $\mathfrak{su}(n)$. Steps 3–4 show all these elements lie in $\mathfrak L$, hence $\mathfrak L=\mathfrak{su}(n)$.

\noindent
\textbf{Step 6 (Synthesis).}
By standard Lie–algebraic controllability (e.g., \cite{DAlessandro2007}), the connected Lie group generated by exponentials of elements in $i\mathcal G$ has Lie algebra $\mathfrak{su}(n)$, hence equals $\mathrm{SU}(n)$. Density plus compactness of $\mathrm{SU}(n)$ yields the stated approximation of any $V\in\mathrm{SU}(n)$ by a finite product of $\exp(-i\theta H)$ with $H\in\mathcal G$ to arbitrary precision $\varepsilon>0$.
\end{proof}


\subsection{On the Moment-operator gap used in the proof}
\label{subsec:moment-gap}

For one random layer \(U\), define the second-moment operator
\[
M_2
\;:=\;
\mathbb{E}\!\left[U^{\otimes 2}\otimes \overline{U}^{\otimes 2}\right].
\]
This is the averaging operator that governs the evolution of second moments under the circuit ensemble. The corresponding Haar benchmark is
\[
M_2^{\mathrm{Haar}}
\;:=\;
\int_{U(d)} V^{\otimes 2}\otimes \overline{V}^{\otimes 2}\, d\mu_{\mathrm{Haar}}(V),
\qquad d=n^m,
\]
which is the projector onto the invariant subspace of the two-copy representation.

In this language, the \emph{moment-operator gap} is the spectral gap of \(M_2\) away from the Haar-fixed subspace. More precisely, if \(\lambda_*\) denotes the largest nontrivial eigenvalue of \(M_2\) in magnitude on the orthogonal complement of the Haar-invariant space, then
\[
\Delta_2
\;:=\;
1-\lambda_*.
\]
Equivalently, when the reduced dynamics is represented by a reversible Markov kernel, one may write
\[
\Delta_2 = 1-\lambda_2,
\]
where \(\lambda_2\) is the second-largest eigenvalue. A positive lower bound on \(\Delta_2\) implies exponential convergence of iterated moments:
\[
\|M_2^L - M_2^{\mathrm{Haar}}\|
\;\le\;
(1-\Delta_2)^L,
\]
up to the perturbative error inherited from the per-block approximate \(2\)-designs.

In the present construction, the local block twirls reduce the relevant second-moment space to the commutant basis
\[
\{I,\mathrm{SWAP}\}^m.
\]
After this reduction, the action of one layer is encoded by an induced Markov kernel on \(\{I,\mathrm{SWAP}\}^m\). Thus, in the proof, the moment-operator gap can be identified with the spectral gap of this reduced kernel, provided the reduction preserves the Haar fixed space and the nontrivial sector. Under connectivity of the block-interaction graph and mild randomness in the entangling layer, this is the standard mechanism by which one obtains inverse-polynomial mixing and hence convergence to an approximate global unitary \(2\)-design \cite{BrandaoHarrowMixing2016,BourgainGamburd2011}. Therefore, more precisely, the proof uses a \emph{second-moment spectral-gap argument} as follows: first identify the fixed space with the Haar second moment, then show that the remaining spectrum is bounded away from \(1\) by at least an inverse polynomial. This gap is exactly what yields the depth bound
\[
L
=
O\!\left(\Delta_2^{-1}\log(1/\varepsilon_{\mathrm{glob}})\right),
\]
up to the approximation error of the local block designs.


\subsection{A derivation of the single-target second moment under a unitary \(2\)-design}
\label{sec:detail-second-moment-single-target}

Let \(D\) be the dimension of the encoded Hilbert space. Fix a normalized input state
\[
\ket{\phi}\in\OH
\]
and a fixed target basis state
\[
\ket{x^\star}\in\OH.
\]
Let
\[
X:=\bigl|\langle x^\star \mid U \mid \phi\rangle\bigr|^2,
\]
where \(U\) is drawn either from Haar measure on \(\mathrm U(D)\) or from an exact unitary \(2\)-design.

Define the rank-one projectors
\[
\rho:=\ket{\phi}\!\bra{\phi},
\qquad
\Pi_{x^\star}:=\ket{x^\star}\!\bra{x^\star}.
\]
Then
\begin{equation}
X
=
\Tr\!\bigl(\Pi_{x^\star}\,U\rho U^\dagger\bigr).
\label{eq:X-trace-form}
\end{equation}

\paragraph{First moment.}
For Haar measure, or any exact unitary \(1\)-design,
\[
\mathbb E_U[U\rho U^\dagger]=\frac{\mathbb I}{D}.
\]
Hence
\begin{align}
\mathbb E[X]
&=
\mathbb E_U\!\left[\Tr\!\bigl(\Pi_{x^\star}U\rho U^\dagger\bigr)\right]
=
\Tr\!\left(\Pi_{x^\star}\,\mathbb E_U[U\rho U^\dagger]\right)
\notag\\
&=
\Tr\!\left(\Pi_{x^\star}\frac{\mathbb I}{D}\right)
=
\frac{\Tr(\Pi_{x^\star})}{D}
=
\frac{1}{D}.
\label{eq:EX-detail}
\end{align}

\paragraph{Second moment.}
Starting from \eqref{eq:X-trace-form},
\[
X^2
=
\Tr\!\bigl(\Pi_{x^\star}U\rho U^\dagger\bigr)^2.
\]
Using the tensor-product trace identity
\[
\Tr(A)\Tr(B)=\Tr(A\otimes B),
\]
we obtain
\begin{align}
X^2
&=
\Tr\!\bigl(\Pi_{x^\star}U\rho U^\dagger\bigr)\,
\Tr\!\bigl(\Pi_{x^\star}U\rho U^\dagger\bigr)
\notag\\
&=
\Tr\!\Bigl((\Pi_{x^\star}\otimes \Pi_{x^\star})
(U\rho U^\dagger)^{\otimes 2}\Bigr).
\label{eq:X2-tensor}
\end{align}
Therefore
\begin{equation}
\mathbb E[X^2]
=
\Tr\!\left(
(\Pi_{x^\star}\otimes \Pi_{x^\star})
\,
\mathbb E_U\bigl[(U\rho U^\dagger)^{\otimes 2}\bigr]
\right).
\label{eq:EX2-start}
\end{equation}

For Haar measure, or any exact unitary \(2\)-design,
\begin{equation}
\mathbb E_U\bigl[(U\rho U^\dagger)^{\otimes 2}\bigr]
=
\frac{\mathbb I+\mathbb F}{D(D+1)},
\label{eq:haar-second-moment}
\end{equation}
where \(\mathbb F\) is the swap operator,
\[
\mathbb F(\ket{a}\otimes\ket{b})=\ket{b}\otimes\ket{a}.
\]
Substituting \eqref{eq:haar-second-moment} into \eqref{eq:EX2-start} gives
\begin{align}
\mathbb E[X^2]
&=
\frac{1}{D(D+1)}
\Tr\!\Bigl((\Pi_{x^\star}\otimes \Pi_{x^\star})(\mathbb I+\mathbb F)\Bigr)
\notag\\
&=
\frac{1}{D(D+1)}
\left[
\Tr(\Pi_{x^\star}\otimes \Pi_{x^\star})
+
\Tr\bigl((\Pi_{x^\star}\otimes \Pi_{x^\star})\mathbb F\bigr)
\right].
\label{eq:EX2-expand}
\end{align}

Now evaluate the two traces separately.

First, since \(\Pi_{x^\star}\) is rank one,
\begin{equation}
\Tr(\Pi_{x^\star}\otimes \Pi_{x^\star})
=
\Tr(\Pi_{x^\star})^2
=
1.
\label{eq:first-trace-rank1}
\end{equation}

Second, use the standard identity
\[
\Tr((A\otimes B)\mathbb F)=\Tr(AB).
\]
With \(A=B=\Pi_{x^\star}\), this gives
\begin{equation}
\Tr\bigl((\Pi_{x^\star}\otimes \Pi_{x^\star})\mathbb F\bigr)
=
\Tr(\Pi_{x^\star}^2)
=
\Tr(\Pi_{x^\star})
=
1.
\label{eq:swap-trace-rank1}
\end{equation}

Substituting \eqref{eq:first-trace-rank1} and \eqref{eq:swap-trace-rank1} into \eqref{eq:EX2-expand}, we obtain
\begin{equation}
\boxed{
\mathbb E[X^2]
=
\frac{2}{D(D+1)}.
}
\label{eq:EX2-final}
\end{equation}

\paragraph{Paley--Zygmund consequence.}
Using \eqref{eq:EX-detail} and \eqref{eq:EX2-final}, Paley--Zygmund gives
\begin{align}
\Pr\!\left[X\ge \frac{1}{2D}\right]
&\ge
\frac{\bigl(\mathbb E[X]-\frac{1}{2D}\bigr)^2}{\mathbb E[X^2]}
\notag\\
&=
\frac{(1/2D)^2}{2/(D(D+1))}
=
\frac{D+1}{8D}
=
\Omega(1).
\label{eq:PZ-single-target-detail}
\end{align}

\paragraph{Approximate \(2\)-design overlap.}
If \(\mathsf U\) is an \(\varepsilon\)-approximate unitary \(2\)-design on the encoded space, then the second-moment operator differs from the Haar value by an error term \(\Delta_\varepsilon\):
\[
\mathbb E_{U\sim\mathsf U}\bigl[(U\rho U^\dagger)^{\otimes 2}\bigr]
=
\frac{\mathbb I+\mathbb F}{D(D+1)}+\Delta_\varepsilon,
\]
with \(\|\Delta_\varepsilon\|\) controlled by the approximation parameter in the norm specified by the design theorem. In that case,
\begin{align}
\mathbb E[X^2]
&=
\Tr\!\left((\Pi_{x^\star}\otimes \Pi_{x^\star})
\left[\frac{\mathbb I+\mathbb F}{D(D+1)}+\Delta_\varepsilon\right]\right)
\notag\\
&=
\frac{2}{D(D+1)}
+
\Tr\!\bigl((\Pi_{x^\star}\otimes \Pi_{x^\star})\Delta_\varepsilon\bigr)
\notag\\
&=
\frac{2}{D(D+1)}\pm O(\varepsilon),
\label{eq:EX2-approx-design}
\end{align}
where the implicit constant depends on the precise definition of approximate \(2\)-design and the norm in which the deviation is controlled. Similarly,
\[
\mathbb E[X]=\frac{1}{D}\pm O(\varepsilon).
\]


\subsection{Rank-\(S\) projector version}
\label{sec:rank-S-second-moment}

The same calculation applies when the rank-one target projector \(\Pi_{x^\star}\) is replaced by a projector \(P\) of rank \(S\), for example the feasible projector \(P_\Pi\).

Define
\[
Y:=\bra{\phi}U^\dagger P U\ket{\phi}
=
\Tr(PU\rho U^\dagger).
\]
Then
\[
\mathbb E[Y]=\frac{S}{D},
\]
and
\begin{align}
\mathbb E[Y^2]
&=
\Tr\!\left((P\otimes P)\frac{\mathbb I+\mathbb F}{D(D+1)}\right)
\notag\\
&=
\frac{\Tr(P\otimes P)+\Tr((P\otimes P)\mathbb F)}{D(D+1)}
\notag\\
&=
\frac{(\Tr P)^2+\Tr(P^2)}{D(D+1)}.
\label{eq:EY2-rankS}
\end{align}
If \(P\) is a projector of rank \(S\), then
\[
\Tr(P)=S,
\qquad
\Tr(P^2)=S,
\]
so
\begin{equation}
\boxed{
\mathbb E[Y^2]
=
\frac{S^2+S}{D(D+1)}.
}
\label{eq:EY2-rankS-final}
\end{equation}
In particular,
\[
\mathbb E[Y^2]
=
\Bigl(\frac{S}{D}\Bigr)^2
\frac{1+1/S}{1+1/D}
\le
2\Bigl(\frac{S}{D}\Bigr)^2
\]


\end{appendix}

\printbibliography

@article{Farhi2014QAOA,
  author = {Farhi, Edward and Goldstone, Jeffrey and Gutmann, Sam},
  title = {A Quantum Approximate Optimization Algorithm},
  journal = {arXiv preprint arXiv:1411.4028},
  year = {2014},
}

@misc{Onahempdata,
  author       = {Onah, Chinonso and Firt, Roman and Michielsen, Kristel},
  title        = {{Dataset: Empirical Quantum Advantage in Constrained Optimization from Encoded Unitary Designs (0.1)}},
  year         = {2025},
  publisher    = {Zenodo},
  doi          = {10.5281/zenodo.15725265},
  url          = {https://doi.org/10.5281/zenodo.15725265},
  note         = {Version 0.1}
}

@article{He2023A,
  title   = {Alignment between initial state and mixer improves QAOA performance for constrained optimization},
  author  = {He, Zichang and Shaydulin, Ruslan and Chakrabarti, Shouvanik and Herman, Dylan and Li, Changhao and Sun, Yue and Pistoia, Marco},
  journal = {npj Quantum Information},
  year    = {2023},
  volume  = {9},
  number  = {1},
  month   = nov,
  issn    = {2056-6387},
  doi     = {10.1038/s41534-023-00787-5},
  url     = {https://doi.org/10.1038/s41534-023-00787-5},
  publisher = {Springer Science and Business Media LLC}
}

@inbook{B_rtschi_2019,
   title={Deterministic Preparation of Dicke States},
   ISBN={9783030250270},
   ISSN={1611-3349},
   DOI={10.1007/978-3-030-25027-0_9},
   booktitle={Fundamentals of Computation Theory},
   publisher={Springer International Publishing},
   author={Bärtschi, Andreas and Eidenbenz, Stephan},
   year={2019},
   pages={126–139},
  %url={http://dx.doi.org/10.1007/978-3-030-25027-0_9},
}

@book{DAlessandro2007,
  author    = {D'Alessandro, Domenico},
  title     = {Introduction to Quantum Control and Dynamics},
  publisher = {CRC Press},
  address   = {Boca Raton, FL},
  year      = {2007},
}

@article{LawlerWood1966BranchBound,
  author  = {Eugene L. Lawler and David E. Wood},
  title   = {Branch-and-Bound Methods: A Survey},
  journal = {Operations Research},
  year    = {1966},
  volume  = {14},
  number  = {4},
  pages   = {699--719},
  doi     = {10.1287/opre.14.4.699}
}

@article{PadbergRinaldi1991BranchCut,
  author  = {Manfred Padberg and Giovanni Rinaldi},
  title   = {A Branch-and-Cut Algorithm for the Resolution of Large-Scale Symmetric Travelling Salesman Problems},
  journal = {SIAM Review},
  year    = {1991},
  volume  = {33},
  number  = {1},
  pages   = {60--100},
  doi     = {10.1137/1033004}
}

@article{LinKernighan1973,
  author  = {Shen Lin and Brian W. Kernighan},
  title   = {An Effective Heuristic Algorithm for the Travelling-Salesman Problem},
  journal = {Operations Research},
  year    = {1973},
  volume  = {21},
  number  = {2},
  pages   = {498--516},
  doi     = {10.1287/opre.21.2.498}
}

@book{GareyJohnson1979,
  title     = {Computers and Intractability: A Guide to the Theory of NP-Completeness},
  author    = {Garey, Michael R. and Johnson, David S.},
  year      = {1979},
  publisher = {W.~H.~Freeman}
}

@incollection{Karp1972,
  author    = {Karp, Richard M.},
  title     = {Reducibility among Combinatorial Problems},
  booktitle = {Complexity of Computer Computations},
  pages     = {85--103},
  year      = {1972},
  publisher = {Plenum}
}

@article{BaeLee2024RecursiveQAOA,
  title   = {Recursive QAOA outperforms the original QAOA for the MAX-CUT problem on complete graphs},
  author  = {Bae, Eunok and Lee, Soojoon},
  journal = {Quantum Information Processing},
  volume  = {23},
  number  = {3},
  pages   = {78},
  year    = {2024},
  doi     = {10.1007/s11128-024-04286-0},
  %url     = {https://doi.org/10.1007/s11128-024-04286-0}
}

@article{Finzgar2024QIRO,
  title     = {Quantum-Informed Recursive Optimization Algorithms},
  author    = {Fin{\v{z}}gar, Jernej Rudi and Kerschbaumer, Aron and Schuetz, Martin J. A. and Mendl, Christian B. and Katzgraber, Helmut G.},
  journal   = {PRX Quantum},
  volume    = {5},
  number    = {2},
  pages     = {020327},
  year      = {2024},
  month     = {5},
  publisher = {American Physical Society},
  doi       = {10.1103/PRXQuantum.5.020327},
  url       = {https://link.aps.org/doi/10.1103/PRXQuantum.5.020327}
}

@book{garey1979computers,
  author    = {Garey, Michael R. and Johnson, David S.},
  title     = {Computers and Intractability: A Guide to the Theory of NP-Completeness},
  publisher = {W.~H.~Freeman and Company},
  year      = {1979},
  address   = {New York}
}

@misc{montanezbarrera2024universalqaoa,
  author       = {Montañez-Barrera, Juan A. and Michielsen, Kristel},
  title        = {Towards a Universal QAOA Protocol: Evidence of a Scaling Advantage in Solving Some Combinatorial Optimization Problems},
  year         = {2024},
  eprint       = {2405.09169},
  archivePrefix= {arXiv},
  primaryClass = {quant-ph},
  note         = {\url{https://arxiv.org/abs/2405.09169}}
}

@software{Qiskit2023,
  author       = {M. Ayrton and J. Home and T. Jones and \emph{et al.}},
  title        = {{Qiskit}: An Open-source Framework for Quantum Computing},
  year         = {2023},
  url          = {https://qiskit.org},
  note         = {Version 0.47}
}

@inproceedings{Osaba2024Qoptlib,
  author    = {Eneko Osaba and Erika Villar-Rodr{\'\i}guez},
  title     = {{QOPTLib}: a Quantum-Computing-Oriented Benchmark for Combinatorial Optimisation Problems},
  booktitle = {Proc.\ \emph{Quantum Tech}. 2024},
  year      = {2024}
}

@article{Chancellor2017,
  author       = {Chancellor, Nicholas},
  title        = {Circuit design for multi-body interactions in superconducting quantum annealing systems with applications to a scalable architecture},
  journal      = {npj Quantum Information},
  year         = {2017},
}

@article{BrandaoHarrowMixing2016,
  author       = {Brand{\~a}o, Fernando G. S. L. and Harrow, Aram W.},
  title        = {Local random quantum circuits are approximate polynomial-designs},
  journal      = {Communications in Mathematical Physics},
  volume       = {346},
  number       = {2},
  pages        = {397--434},
  year         = {2016},
}

@article{Egger2021Warm,
  author       = {Egger, Daniel J. and Gambella, Claudia and Tomesh, Tyler and Woerner, Stefan},
  title        = {Warm-starting quantum optimization},
  journal      = {PRX Quantum},
  volume       = {2},
  pages        = {040348},
  year         = {2021},
}

@misc{carmo2025warmstarting,
  author       = {do Carmo, Rafael S. and Santana, Marcos C. S. and Fanchini, Felipe F. and de Albuquerque, Victor H. C. and Papa, Jo\~ao P.},
  title        = {Warm-Starting QAOA with XY Mixers: A Novel Approach for Quantum-Enhanced Vehicle Routing Optimization},
  year         = {2025},
  eprint       = {2504.19934},
  archivePrefix= {arXiv},
  primaryClass = {quant-ph},
}

@article{smith-miles2025tsp-ntqa,
  author       = {Smith{-}Miles, Kate A. and Hoos, Holger H. and Wang, Hao and B{\"a}ck, Thomas H. W. and Osborne, Tobias J.},
  title        = {The travelling salesperson problem and the challenges of near-term quantum advantage},
  journal      = {Quantum Science and Technology},
  year         = {2025},
  volume       = {10},
  number       = {3},
  pages        = {033001},
  doi          = {10.1088/2058-9565/add61d},
  month        = {5},
  publisher    = {IOP Publishing},
  %url          = {https://doi.org/10.1088/2058-9565/add61d}
}

@article{PerezSalinasWangBonetMonroig2024,
  author       = {P{\'e}rez-Salinas, A. and Wang, H. and Bonet-Monroig, X.},
  title        = {Analyzing variational quantum landscapes with information content},
  journal      = {npj Quantum Information},
  year         = {2024},
  volume       = {10},
  pages        = {27},
  month        = feb,
  doi          = {10.1038/s41534-024-00819-8},
  %url          = {https://doi.org/10.1038/s41534-024-00819-8}
}

@article{HadfieldHoggRieffel2022,
  title     = {Analytical framework for quantum alternating operator ans{\"a}tze},
  author    = {Hadfield, Stuart and Hogg, Tad and Rieffel, Eleanor G.},
  journal   = {Quantum Science and Technology},
  publisher = {IOP Publishing},
  year      = {2022},
  month     = dec,
  volume    = {8},
  number    = {1},
  pages     = {015017},
  doi       = {10.1088/2058-9565/aca3ce},
 % url       = {http://dx.doi.org/10.1088/2058-9565/aca3ce},
  issn      = {2058-9565}
}

@article{BengioLodiProuvost2021,
  author  = {Yoshua Bengio and Andrea Lodi and Antoine Prouvost},
  title   = {Machine Learning for Combinatorial Optimization: A Methodological Tour d'Horizon},
  journal = {European Journal of Operational Research},
  year    = {2021},
  volume  = {290},
  number  = {2},
  pages   = {405--421},
  doi     = {10.1016/j.ejor.2020.07.063},
  %url     = {https://doi.org/10.1016/j.ejor.2020.07.063},
}

@INPROCEEDINGS{onah2025waas,
  author={Onah, Chinonso and Misciasci, Neel and Othmer, Carsten and Michielsen, Kristel},
  booktitle={2025 IEEE International Conference on Quantum Computing and Engineering (QCE)}, 
  title={QUEST: QUantum-Enhanced Shared Transportation}, 
  year={2025},
  volume={01},
  number={},
  pages={2149-2160},
  doi={10.1109/QCE65121.2025.00235}}

@misc{diker2022,
      title={Deterministic construction of arbitrary $W$ states with quadratically increasing number of two-qubit gates}, 
      author={Firat Diker},
      year={2022},
      eprint={1606.09290},
      archivePrefix={arXiv},
      primaryClass={quant-ph},
      url={https://arxiv.org/abs/1606.09290}, 
}

@article{bravyi_gosset_2016,
  author  = {S. Bravyi and D. Gosset},
  title   = {Improved classical simulation of quantum circuits dominated by Clifford gates},
  journal = {Phys. Rev. Lett.},
  volume  = {116},
  pages   = {250501},
  year    = {2016}
}

@book{PapadimitriouSteiglitz1982,
  author    = {Christos H. Papadimitriou and Kenneth Steiglitz},
  title     = {Combinatorial Optimization: Algorithms and Complexity},
  publisher = {Prentice Hall},
  year      = {1982}
}

@book{GodsilRoyle2001,
  author    = {Chris Godsil and Gordon Royle},
  title     = {Algebraic Graph Theory},
  series    = {Graduate Texts in Mathematics},
  volume    = {207},
  publisher = {Springer},
  address   = {New York},
  year      = {2001}
}

@article{tsvelikhovskiy2024symmetries,
  title={Symmetries and Dimension Reduction in Quantum Approximate Optimization Algorithm},
  author={Tsvelikhovskiy, B. and Safro, I. and Alexeev, Y.},
  journal={arXiv preprint arXiv:2309.13787},
  year={2023},
  archivePrefix={arXiv},
  eprint={2309.13787},
  primaryClass={quant-ph},
  %version={2},
  %url={https://arxiv.org/abs/2309.13787}
}

@article{tsvelikhovskiy2024equivariant,
  title={Equivariant QAOA and the Duel of the Mixers},
  author={Tsvelikhovskiy, B. and Safro, I. and Alexeev, Y.},
  journal={arXiv:2405.07211},
  year={2024},
  archivePrefix={arXiv},
  eprint={2405.07211},
  primaryClass={quant-ph},
  %version={1},
  %url={https://arxiv.org/abs/2405.07211}
}

@article{Lucas2014Ising,
  author    = {Andrew Lucas},
  title     = {Ising formulations of many {NP} problems},
  journal   = {Frontiers in Physics},
  year      = {2014},
  volume    = {2},
  pages     = {5},
  doi       = {10.3389/fphy.2014.00005}
}

@article{Edmonds1965Blossom,
  author    = {Jack Edmonds},
  title     = {Paths, Trees, and Flowers},
  journal   = {Canadian Journal of Mathematics},
  year      = {1965},
  volume    = {17},
  pages     = {449--467},
  doi       = {10.4153/CJM-1965-045-4}
}

@article{Cerezo2021VQAReview,
  title   = {Variational quantum algorithms},
  author  = {Cerezo, M. and Arrasmith, A. and Babbush, R. and Benjamin, S. C. and Endo, S. and Fujii, K. and McClean, J. R. and Mitarai, K. and Yuan, X. and Cincio, L. and Coles, P. J.},
  journal = {Nature Reviews Physics},
  volume  = {3},
  number  = {9},
  pages   = {625--644},
  year    = {2021},
  doi     = {10.1038/s42254-021-00348-9}
}

@article{Tilly2022VQEReview,
  title   = {The Variational Quantum Eigensolver: A review of methods and best practices},
  author  = {Tilly, J. and Chen, H. and Cao, S. and Picozzi, D. and Setia, K. and Grant, E. and Wossnig, L. and Rungger, I. and Booth, G. H. and Tennyson, J.},
  journal = {PRX Quantum},
  volume  = {3},
  number  = {3},
  pages   = {030204},
  year    = {2022},
  doi     = {10.1103/PRXQuantum.3.030204}
}

@article{McClean2018BarrenPlateaus,
  title   = {Barren plateaus in quantum neural network training landscapes},
  author  = {McClean, Jarrod R. and Boixo, Sergio and Smelyanskiy, Vadim N. and Babbush, Ryan and Neven, Hartmut},
  journal = {Nature Communications},
  volume  = {9},
  number  = {1},
  pages   = {4812},
  year    = {2018},
  doi     = {10.1038/s41467-018-07090-4}
}

@article{BourgainGamburd2011,
  author    = {Jean Bourgain and Alex Gamburd},
  title     = {A Spectral Gap Theorem in {$\mathrm{SU}(d)$}},
  journal   = {J. Eur. Math. Soc.},
  volume    = {14},
  number    = {5},
  pages     = {1455--1511},
  year      = {2012}
}

@inproceedings{Li2021CoDesign,
  author    = {Li, Gushu and Ding, Yufei and Xie, Yuan},
  title     = {On the Co-Design of Quantum Software and Hardware},
  booktitle = {ICCAD '21: IEEE/ACM International Conference on Computer-Aided Design},
  year      = {2021},
  doi       = {10.1145/3477206.3477464}
}

@article{Tomesh2021QuantumCodesign,
  author  = {Tomesh, Teague and Martonosi, Margaret},
  title   = {Quantum Codesign},
  journal = {IEEE Micro},
  year    = {2021},
  volume  = {41},
  number  = {5},
  pages   = {33--40},
  doi     = {10.1109/MM.2021.3094461}
}

@article{Hadfield2019AOA,
  author = {Hadfield, Stuart and Wang, Zhihui and O'Gorman, Bryan and Rieffel, Eleanor G. and Venturelli, Davide and Biswas, Rupak},
  title  = {From the Quantum Approximate Optimization Algorithm to a Quantum Alternating Operator Ansatz},
  journal= {Algorithms},
  year   = {2019},
  volume = {12},
  number = {2},
  pages  = {34},
  doi    = {10.3390/a12020034}
}

@article{Wang2020XYMixers,
  author  = {Wang, Zhihui and Hadfield, Stuart and Jiang, Zhenyu and Rieffel, Eleanor G.},
  title   = {$XY$ mixers: Analytical and numerical results for the quantum approximate optimization algorithm},
  journal = {Physical Review A},
  year    = {2020},
  volume  = {101},
  number  = {1},
  pages   = {012320},
  doi     = {10.1103/PhysRevA.101.012320}
}

@article{Vidal2003Efficient,
  author  = {Vidal, Guifr{\'e}},
  title   = {Efficient Classical Simulation of Slightly Entangled Quantum Computations},
  journal = {Phys. Rev. Lett.},
  volume  = {91},
  number  = {14},
  pages   = {147902},
  year    = {2003},
  doi     = {10.1103/PhysRevLett.91.147902}
}

@book{Lawler1985TSP,
  title     = {The Traveling Salesman Problem: A Guided Tour of Combinatorial Optimization},
  author    = {Lawler, Eugene L. and Lenstra, Jan Karel and Rinnooy Kan, Alexander H. G. and Shmoys, David B.},
  publisher = {Wiley},
  year      = {1985}
}

@article{GareyJohnsonSethi1976,
  author  = {Garey, Michael R. and Johnson, David S. and Sethi, Ravi},
  title   = {The Complexity of Flowshop and Jobshop Scheduling},
  journal = {Mathematics of Operations Research},
  year    = {1976},
  volume  = {1},
  number  = {2},
  pages   = {117--129},
  doi     = {10.1287/moor.1.2.117}
}

@article{Koopmans1957QAP,
  author  = {Koopmans, Tjalling C. and Beckmann, Martin},
  title   = {Assignment Problems and the Location of Economic Activities},
  journal = {Econometrica},
  year    = {1957},
  volume  = {25},
  number  = {1},
  pages   = {53--76},
  doi     = {10.2307/1907742}
}

@article{Loiola2007QAPSurvey,
  author  = {Loiola, E. M. and de Abreu, N. M. M. and Boaventura-Netto, P. O. and Hahn, P. and Querido, T.},
  title   = {A Survey for the Quadratic Assignment Problem},
  journal = {European Journal of Operational Research},
  year    = {2007},
  volume  = {176},
  number  = {2},
  pages   = {657--690},
  doi     = {10.1016/j.ejor.2005.09.032}
}

@book{TothVigo2014VRP,
  author    = {Toth, Paolo and Vigo, Daniele},
  title     = {Vehicle Routing: Problems, Methods, and Applications},
  edition   = {2},
  series    = {MOS-SIAM Series on Optimization},
  publisher = {SIAM},
  year      = {2014},
  doi       = {10.1137/1.9781611973594}
}

@article{SahniGonzalez1976GAP,
  author  = {Sahni, Sartaj and Gonzalez, Teofilo},
  title   = {P-Complete Approximation Problems},
  journal = {Journal of the ACM},
  year    = {1976},
  volume  = {23},
  number  = {3},
  pages   = {555--565},
  doi     = {10.1145/321958.321975}
}

@book{MartelloToth1990Knapsack,
  author    = {Martello, Silvano and Toth, Paolo},
  title     = {Knapsack Problems: Algorithms and Computer Implementations},
  publisher = {Wiley},
  year      = {1990}
}

@article{LenstraRinnooyKan1977,
  author  = {Lenstra, Jan Karel and Rinnooy Kan, Alexander H. G. and Brucker, Peter},
  title   = {Complexity of Machine Scheduling Problems},
  journal = {Annals of Discrete Mathematics},
  year    = {1977},
  volume  = {1},
  pages   = {343--362},
  doi     = {10.1016/S0167-5060(08)70743-X}
}

@article{Kuhn1955Hungarian,
  author  = {Kuhn, Harold W.},
  title   = {The Hungarian Method for the Assignment Problem},
  journal = {Naval Research Logistics Quarterly},
  year    = {1955},
  volume  = {2},
  number  = {1-2},
  pages   = {83--97},
  doi     = {10.1002/nav.3800020109}
}

@article{Munkres1957Hungarian,
  author  = {Munkres, James},
  title   = {Algorithms for the Assignment and Transportation Problems},
  journal = {Journal of the Society for Industrial and Applied Mathematics},
  year    = {1957},
  volume  = {5},
  number  = {1},
  pages   = {32--38},
  doi     = {10.1137/0105003}
}

@book{Applegate2007,
  author    = {Applegate, David L. and Bixby, Robert E. and Chv{\'a}tal, Vasek and Cook, William J.},
  title     = {The Traveling Salesman Problem: A Computational Study},
  publisher = {Princeton University Press},
  year      = {2007},
  address   = {Princeton, NJ}
}

@article{Bellman1962,
  author  = {Bellman, Richard},
  title   = {Dynamic Programming Treatment of the Traveling Salesman Problem},
  journal = {Journal of the ACM},
  year    = {1962},
  volume  = {9},
  number  = {1},
  pages   = {61--63},
  doi     = {10.1145/321105.321111}
}

@article{Fuchs2022ConstrainedMixers,
  author  = {Fuchs, Franz G. and Bassa, Ruben Pariente},
  title   = {Constraint Preserving Mixers for the Quantum Approximate Optimization Algorithm},
  journal = {Algorithms},
  year    = {2022},
  volume  = {15},
  number  = {6},
  pages   = {202},
  doi     = {10.3390/a15060202}
}

@article{LaRose2022MixerPhaser,
  author  = {LaRose, Ryan and Cerezo, M. and Czarnik, Piotr and others},
  title   = {Mixer-phaser Ans{\"a}tze for quantum optimization with hard constraints},
  journal = {Quantum Machine Intelligence},
  year    = {2022},
  volume  = {4},
  number  = {2},
  pages   = {12},
  doi     = {10.1007/s42484-022-00069-x}
}

@misc{BaertschiEidenbenz2020,
  author  = {B{\"a}rtschi, Andreas and Eidenbenz, Stephan},
  title   = {Grover Mixers for QAOA: Shifting Complexity from Mixer Design to State Preparation},
  howpublished = {arXiv:2006.00354},
  year    = {2020},
  doi     = {10.48550/arXiv.2006.00354}
}

@misc{Xie2024CVRP,
  author  = {Xie, Ningyi and Lau, Hoong Chuin},
  title   = {A Feasibility-Preserved Quantum Approximate Solver for the Capacitated Vehicle Routing Problem},
  howpublished = {arXiv:2308.08785},
  year    = {2024},
  url     = {https://arxiv.org/abs/2308.08785}
}

@misc{doCarmo2025warmstartingqaoa,
  author       = {do Carmo, Rafael S. and Santana, Márcio C. S. and Fanchini, Felipe F. and de Albuquerque, Vinícius H. C. and Papa, João P.},
  title        = {Warm-Starting QAOA with XY Mixers: A Novel Approach for Quantum-Enhanced Vehicle Routing Optimization},
  howpublished = {arXiv preprint arXiv:2504.19934},
  year         = {2025},
  eprint       = {2504.19934},
  archivePrefix= {arXiv},
  primaryClass = {quant-ph}
}
\end{document}